\documentclass[12pt,reqno]{amsart}
\usepackage{amsmath}
\usepackage{amssymb}
\usepackage[left=3cm,top=3cm,right=3cm,bottom=3cm]{geometry}
\usepackage{epsfig}
\usepackage[normalem]{ulem}
\usepackage[usenames,dvipsnames]{color}
\usepackage{todonotes}
\usepackage[colorlinks, citecolor=blue, linkcolor=blue, urlcolor=blue]{hyperref}

\usepackage{tikz}
\usetikzlibrary{backgrounds}
\usetikzlibrary{intersections}
\usepackage{floatrow}

\newtheorem{theorem}{Theorem}[section]
\newtheorem{lemma}[theorem]{Lemma}

\newcommand\eps{\varepsilon}
\newcommand{\E}{\mathbb E}
\newcommand{\Cov}{\mathbb C\textrm{ov}}
\newcommand{\Var}{\mathbb V\textrm{ar}}
\newcommand{\Prob}{\mathbb{P}}
\newcommand{\Bin}{\mathrm{Bin}}
\newcommand{\Geom}{\mathrm{Geom}}
\newcommand{\Nn}{{\mathbb N}}

\title{On Broadcasting Time in the Model of Travelling Agents}
\author[R.~Huq \and B.~Kami\'nski \and A.~Mashatan \and P.~Pra\l{}at \and P.~Szufel]
{Reaz Huq \and Bogumi\l{} Kami\'nski \and Atefeh Mashatan \and Pawe\l{} Pra\l{}at \and Przemys\l{}aw Szufel}

\address{SGH Warsaw School of Economics, Warsaw, Poland}
\email{bogumil.kaminski@sgh.waw.pl, pszufe@sgh.waw.pl}

\address{Ryerson University, Toronto, Canada}
\email{reaz.huq@ryerson.ca, amashatan@ryerson.ca, pralat@ryerson.ca}

\date{\today}

\begin{document}

\maketitle

\begin{abstract}
Consider the following broadcasting process run on a connected graph $G=(V,E)$. Suppose that $k \ge 2$ agents start on vertices selected from $V$ uniformly and independently at random. One of the agents has a message that she wants to communicate to the other agents. All agents perform independent random walks on $G$, with the message being passed when an agent that knows the message meets an agent that does not know the message. The broadcasting time $\xi(G,k)$ is the time it takes to spread the message to all agents.

Our ultimate goal is to gain a better understanding of the broadcasting process run on real-world networks of roads of large cities that might shed some light on the behaviour of future autonomous and connected vehicles. Due to the complexity of road networks, such phenomena have to be studied using simulation in practical applications. In this paper, we study the process on the simplest scenario, i.e., the family of complete graphs, as in this case the problem is analytically tractable. We provide tight bounds for $\xi(K_n,k)$ that hold asymptotically almost surely for the whole range of the parameter $k$. These theoretical results reveal interesting relationships and, at the same time, are also helpful to understand and explain the behaviour we observe in more realistic networks.
\end{abstract}

\section{Introduction and Motivation}

In this paper, we investigate the problem of broadcasting messages between agents that randomly move on a connected graph $G=(V,E)$. The assumption is that $k \ge 2$ agents start the process at random locations on the graph and then perform a random walk along its vertices. One agent, selected in advance, initially possesses some information. If two agents meet at some point during the process and only one of them possesses the information, it is passed along to the other agent. The broadcasting time $\xi(G,k)$ is the time it takes to spread the message to all agents. (Formal definition will be provided in Section~\ref{sec:problem}.)

The motivation for this line of work stems from the need to better understand the process of broadcasting messages between cars that are connected via Dedicated Short-Range Communication (DSRC) devices. DSRC devices are a component of Vehicle-to-Vehicle (V2V) communication infrastructure \cite{Kenney, Bai}. V2V solutions can be combined with cellular networks (Vehicle-to-Infrastructure, V2I) to form an efficient communication solution for a transportation system \cite{Dey}. For scenarios involving direct communication within a city, dedicated protocols, such as the Urban Vehicular BroadCAST (UV-CAST) protocol \cite{Viriyasitavat} have been developed. In UV-CAST the authors use agent-based simulation of a real city to test the protocol's efficiency.

\medskip

In this paper, we focus solely on understanding how the density of vehicles in a network affects the efficiency of broadcasting messages in the theoretical case of a complete graph, as this is a scenario that can be investigated analytically.
We provide a complete characterization of $\xi(K_n,k)$ for the whole range of the parameter $k$. Interestingly, $\xi(K_n,k)$ is well concentrated around $2n \ln k / k$ for a wide range of possible values of $k$, but the behaviour changes when $k$ is very large, namely, when $k$ is linear in $n$ (see Theorem~\ref{thm:main} for details).

For more realistic scenarios involving real-world networks of roads, in~\cite{Lukasz} we have developed a general simulation framework. An in-depth simulation-approach to the broadcasting problem will be the subject of an accompanying paper~\cite{simulations}.
Selected results of the experiments performed in~\cite{simulations} are given in Appendix~\ref{sec:appendix}. Despite the fact that real-world road networks have completely different structure than complete graphs, simulations led us to the hypothesis that the broadcasting time in real scenario is still proportional to $n \ln k / k$, where $n$ is the number of vertices in the graph (which correspond to intersections) and $k$ is the number of traveling agents (which correspond to cars).

\medskip

The paper is structured as follows. In the next section, Section~\ref{sec:related_problems}, we discuss related problems. We formally define our problem, introduce asymptotic notation, and state the main result in Section~\ref{sec:problem}. The whole of Section~\ref{sec:proofs} is devoted to proving the main result. In Appendix~\ref{sec:appendix}, we provide some initial results on real-world road networks that were our original motivation for the study presented in this paper. All source code for numerical experiments and results of the simulations performed for this paper are available on GitHub\footnote{\url{https://github.com/pszufe/SignalBroadcastingSim.jl}}.

\section{Related Problems}\label{sec:related_problems}

The performance of a random walk in a networks is a fundamental process that has found applications in many areas of computer science.
Since this paper contains theoretical results, let us concentrate on rigorous, theoretical results of related processes.
As this is still a very broad topic, we only scratch the surface and focus on multiple random walks performed simultaneously (which has many applications in distributed computing, such as sampling) and processes run on complete graphs. For more on other directions, we direct the reader to one of the many books on Markov chains; see, for example~\cite{book_Yuval}.

Suppose there are $k \ge 2$ particles, each making a simple random walk on a graph $G$. Even if the particles are oblivious of each other, it is important and non-trivial to estimate the (vertex) cover time, an extensively studied graph parameter that is defined as the expected time required for the process to visit every vertex of $G$. Questions become more interesting (and difficult) once we allow particles to interact once they meet. We assume that interaction occurs only when meeting at a vertex, and that the random walks made by the particles are otherwise independent. There are at least four interesting variants of this process:
\begin{itemize}
    \item \emph{Predator-Prey}: estimate the expected time-to-extinction of the prey particles under the assumption that $k$ predator and $\ell$ prey particles walk independently; predators eat prey particles upon meeting at a vertex.
    \item \emph{Coalescing particles}: estimate the expected time to coalesce to a single particle under the assumption that $k$ particles walk independently and coalesce upon meeting at a vertex.
    \item \emph{Annihilating particles}: estimate the expected time-to-extinction of all particles under the assumption that $k = 2\ell$ particles walk independently and destroy each other (pairwise) upon meeting at a vertex.
    \item \emph{Talkative particles}: estimate the expected time to broadcast a message---this is exactly the problem that we are concerned with in this paper.
\end{itemize}
All of these variants have been studied for random $d$-regular graphs $\mathcal{G}_{n,d}$~\cite{Alan_d-reg}. In particular, if $d \ge 3$ is a fixed constant and $k \le n^{\eps}$ for a sufficiently small constant $\eps>0$, then asymptotically almost surely (see the next section for a definition and a notation used)
$$
\xi(\mathcal{G}_{n,d}, k) \sim \frac {2 H_{k-1}}{k} \cdot \frac {d-1}{d-2} \cdot n.
$$
This is the only theoretical result on the broadcasting time we are aware of.

On the other hand, the variant of coalescing particles is very well-studied, mainly because of its surprising connection to the voter model~\cite{voter2,voter6}. The state of the process at a given time $t$ is described by a function $\mu_t : V(G) \to O$, where $V(G)$ is the vertex set of a graph $G$ and $O$ is a given set of possible opinions. Each vertex $v\in V(G)$ ``wakes up'' at rate 1 and when it happens at a time $t>0$, $v$ chooses one of its neighbours $w$ uniformly at random and updates its value $\mu_t(v)$ to the opinion of $w$; all other opinions remain the same. A classical duality result (see, for example,~\cite{voter2,voter6}) directly relates the state of the process at a given time to a system of coalescing random walks on $G$ moving backwards in time. As we already mentioned, there are many interesting results on the coalescing time. Let us only mention a beautiful conjecture posed by Aldous and Fill in the mid-nineties (Open problem 13, Chapter~14 of~\cite{voter2}). They conjectured an upper bound for the mean coalescent time in terms of the mean hitting time of a single random walk. The conjecture was proved in~\cite{Roberto12,Roberto13}.

\medskip

Let us now briefly discuss the following well-known and well-studied rumour spreading protocols: \emph{Push} and \emph{Push\&Pull}. Suppose that one vertex in a network is aware of a piece of information, the ``rumour'', and wants to spread it to all vertices. In each round of the \emph{Push} protocol, every informed vertex contacts a random neighbour and sends the rumour to it (``pushes'' the rumour). In \emph{Push\&Pull}, uninformed vertices can also contact a random neighbour to get the rumour if the neighbour knows it (``pulls'' the rumour).

There is a long sequence of interesting and important papers studying the runtime of \emph{Push} on the complete graph. The first paper considering this protocol is~\cite{Push10} but more precise bounds were provided in~\cite{Push17} and then in~\cite{Push5}, in which it was shown that the process can essentially be stochastically bounded (from both sides) by coupon collector type problems. A very recent paper~\cite{Latin} answers some remaining questions; it both determines the limiting distribution and explains why it is a difficult problem: the runtime, scaled appropriately by $(\log_2 n + \ln n)$, has no limiting distribution; instead, it exhibits a double-oscillatory behaviour. \emph{Push} has been extensively studied on several other graph classes besides complete graphs.

The \emph{Push\&Pull} protocol has an equally long sequence of interesting papers studying it. The synchronous version of the protocol (as described above) was introduced in~\cite{Pull5} and popularized in~\cite{Pull23}. However, such synchronized models (that is, models in which all vertices take action simultaneously at discrete time steps) are not plausible for many applications, including real-world social networks. As a result, an asynchronous version of the model with a continuous timeline was introduced in~\cite{Pull4}. In this variant, each vertex has its own independent clock that rings at the times of a rate 1 Poisson process with the protocol specifying what a vertex has to do when its own clock rings. The first theoretical relationships between the spread times in the two variants was provided in~\cite{Wormald}.

\section{Formulation of the Problem and the Main Result}\label{sec:problem}

In this section, we formally define the process we aim to analyze (Subsection~\ref{sec:proglem}). We define it for any connected graph but in this paper we focus on complete graphs. As our results are asymptotic in nature, we need to introduce the asymptotic notation that is used throughout the entire paper (Subsection~\ref{sec:asymptotic}). Finally, we state the main result that combines all ranges for the number of agents involved (Subsection~\ref{sec:main_result}).

\subsection{Problem}\label{sec:proglem}

Suppose that we are given a connected graph $G=(V,E)$ on $n=|V|$ vertices. Let $k \ge 2$ be any natural number; $k=k(n)$ may be a function of $n$ that tends to infinity as $n\to \infty$. There are $k$ agents, one of which is green and the remaining ones are white. The process starts at round $t=0$ with agents located randomly on vertices of $G$; that is, each agent starts at any vertex $v \in V$ with probability $1/n$, independently of other agents and independently of her colour. Each agent synchronously performs an independent random walk, regardless if she is green or white. In other words, an agent occupying vertex $v \in V$ moves to any neighbour of $v$ with probability equal to $1/\deg(v)$. When green and white agents meet at some \emph{round} $t \ge 0$, the white agent becomes green. In particular, all agents that start at the same vertex as the initial green agent are green from the very beginning. Let $\xi = \xi(G, k)$ be the time it takes for all agents to become green. (Note that $\xi$ is a random variable even if $G$ is a deterministic graph.)

\medskip

We say that the process is at \emph{phase} $\ell$ ($1 \le \ell \le k$) if there are $\ell$ green agents (and so $k-\ell$ white agents). The first phase is usually phase $1$, unless some white agents start at the same vertex as the green agent but this is rare if $k$ is small. Clearly, the process always moves from a smaller phase to a larger phase. If $k$ is small, then typically it takes a lot of rounds for the process to move to another phase but some phases may be skipped. If $k$ is large, then skipping phases is quite common. The process ends at the end of round $\xi$, when we are about to move to phase $k$.

\subsection{Asymptotic Notation}\label{sec:asymptotic}

Our results are asymptotic in nature, that is, we will assume that $n\to\infty$. We are interested in events that hold \emph{asymptotically almost surely} (\emph{a.a.s.}), that is, events that hold with probability tending to 1 as $n\to \infty$.

Given two functions $f=f(n)$ and $g=g(n)$, we will write $f=O(g)$ if there exists an absolute constant $c$ such that $f \leq cg$ for all $n$, $f=\Omega(g)$ if $g=O(f)$, $f=\Theta(g)$ if $f=O(g)$ and $f=\Omega(g)$, and we write $f=o(g)$ or $f\ll g$ if the limit $\lim_{n\to\infty} f/g=0$. In addition, we write $f\gg g$ if $g=o(f)$ and we write $f\sim g$ if $f=(1+o(1))g$.

\subsection{Main Result}\label{sec:main_result}

Our ultimate goal is to analyze the process for the real network associated with some large city. We present preliminary results of such a study in Appendix~\ref{sec:appendix}. In this paper, we try to gain some intuition for and experience with the process by considering it for complete graphs which, perhaps surprisingly, is quite a challenging and tedious task. It would also make sense to investigate grids, binomial random graphs, and random geometric graphs, both theoretically and via simulations.

\bigskip

Let us summarize the main results in one theorem. More detailed statements can be found in the next section.

\begin{theorem}\label{thm:main}
Depending on the parameter $k=k(n)$, the following properties hold:
\begin{itemize}
\item [(a)] If $k=O(1)$, then a.a.s.
$$
\frac {n \ln k}{\omega k} \le \xi(K_n,k) \le \frac {\omega n \ln k}{k},
$$
where $\omega=\omega(n)$ is any function tending to infinity as $n \to \infty$.
Moreover, $\E [\xi(K_n,k)] \sim 2nH_{k-1}/k$, where $H_{k-1} = \sum_{\ell=1}^{k-1} \ell^{-1}$ is the Harmonic number.
\item [(b)] If $1 \ll k \ll n$, then a.a.s.
$$
\xi(K_n,k) \sim \frac {2 n \ln k}{k}.
$$
\item [(c)] If $k \sim cn$ for some $c \in (0,\infty)$, then a.a.s.
$$
\xi(K_n,k) \sim \left(\frac {1}{c} + \frac {1}{\ln(1+c)} \right) \ln n.
$$
\item [(d)] If $k \sim cn$ for some $c=c(n)$ such that $1 \ll c = n^{o(1)}$, then a.a.s.
$$
\xi(K_n,k) \sim \frac {\ln n}{\ln c}.
$$
\item [(e)] If $k = n^{1 + x + o(1)}$ for some $x \in ( \frac {1}{i}, \frac {1}{i-1})$, $i \in \Nn \setminus \{1,2\}$, then a.a.s.
$$
\xi(K_n,k) = i.
$$
\item [(f)] If $k = n^{1 + x + o(1)}$ for some $x \in ( \frac {1}{2}, \infty)$, then a.a.s.
$$
\xi(K_n,k) = 2.
$$
\item [(g)] If $k = n^{1 + 1/i + o(1)}$ for some $i \in \Nn \setminus \{1\}$, then a.a.s.
$$
\xi(K_n,k) =
\begin{cases}
i+1 & \textrm{ if } (k/n)^i < (1-\eps) n \ln n \textrm{ for some } \eps > 0 \\
i & \textrm{ if } (k/n)^i > (1+\eps) n \ln n \textrm{ for some } \eps > 0 \\
i \textrm{ or } i+1 & \textrm{ otherwise.}
\end{cases}
$$
\end{itemize}
\end{theorem}

\section{Proofs}\label{sec:proofs}

This whole section is devoted to proving Theorem~\ref{thm:main}. We investigate the process running on $K_n$, the complete graph on $n$ vertices. Depending on the number of agents involved (parameter $k=k(n)$), the proof requires different approaches. We will deal with each subrange of $k$ independently. However, before we start, let us state Chernoff's bound, a well-known concentration inequality that we will use often.

\subsection{Concentration inequalities}

Throughout the paper, we will be using the following concentration inequality. Let $X \in \textrm{Bin}(n,p)$ be a random variable with the binomial distribution with parameters $n$ and $p$. Then, a consequence of Chernoff's bound (see e.g.~\cite[Corollary~2.3]{JLR}) is that
\begin{equation}\label{eq:chern}
\Prob( |X-\E X| \ge \eps \E X) \le 2\exp \left( - \frac {\eps^2 \E X}{3} \right)
\end{equation}
for  $0 < \eps < 3/2$.
Moreover, let us mention that the bound holds for the general case in which $X=\sum_{i=1}^n X_i$ and $X_i \in \textrm{Bernoulli}(p_i)$ with (possibly) different $p_i$ (again, e.g.~see~\cite{JLR} for more details).

\subsection{Chebyshev's inequality}

At one point, we will need to use the classic Chebyshev's inequality. It can be applied in a more general scenario but here we only present a specific case that suffices for our application. Let $X$ be any random variable taking values from the set of non-negative integers. Then, for any $\eps > 0$,
\begin{equation}\label{eq:cheb}
\Prob( |X-\E X| \ge \eps \E X) \le \frac {\Var X}{(\eps \E X)^2}.
\end{equation}
In particular, if $\Var X = o( (\E X)^2 )$, then a.a.s.\ $X \sim \E[X]$.

\subsection{Concentration for Case (a) and a subrange of Case (b)}
We start with the case of $k$ being relatively small, that is, $2 \le k \ll (n/\log n)^{1/3}$.

\begin{theorem}\label{thm:Kn:sparse}
Let $\omega_0=\omega_0(n)$ be any function tending to infinity as $n \to \infty$. Then the following properties hold a.a.s.
\begin{itemize}
\item If $k=O(1)$, then
$$
\frac {n \ln k}{\omega_0 k} \le \xi(K_n,k) \le \frac {\omega_0 n \ln k}{k}.
$$
\item If $1 \ll k \ll (n/\ln n)^{1/3}$, then
$$
\xi(K_n,k) = \frac {2 n \ln k}{k} \left( 1 + O\left( \frac {1}{(\ln k)^{1/3}} \right) \right) \sim \frac {2 n \ln k}{k}.
$$
\end{itemize}
\end{theorem}
\begin{proof}
Let $\omega=\omega(n)$ be any function tending to infinity as $n \to \infty$ (sufficiently slowly so that some bounds claimed below hold and $\omega(n) \le \omega_0(n)$ for all $n$). Suppose that
$$
k \le \frac {1}{\omega} \left( \frac {n}{\ln n} \right)^{1/3}.
$$
Since we aim for upper bounds for $\xi(K_n,k)$ that are at most
$$
N := n \ \frac {\omega \ln k}{k},
$$
we may only consider $N$ rounds of the process. Formally, we stop the process prematurely if the process is \emph{not} finished by the end of round $N$, and in such cases no upper bound is claimed. However, we will show that a.a.s.\ the process does terminate before time $N$ and so the desired bounds are established.

\medskip

Note that the probability that the process starts with no white agents occupying the same vertex as the initial green agent is equal to
$$
\left( 1 - \frac {1}{n} \right)^{k-1} = \exp \left( -\frac {k-1}{n} + O \left( \frac {k}{n^2} \right) \right) \sim 1,
$$
since $k = o((n/\ln n)^{1/3}) = o(n)$. Hence, since we aim for a result that holds a.a.s., we may assume that the process starts at phase 1.

\medskip

In the original definition of the problem agents move simultaneously but it is more convenient to split each round $t$ into two sub-rounds. Let us first move green agents and, once they land on the corresponding vertices, we move white agents. We say that a given round is \emph{unusual} if at least one of the following two events hold during the first sub-round:
$U_1(t)$---two green agents meet,
$U_2(t)$---some green agent moves to a vertex occupied by some white agent.
Let $U(t)=U_1(t) \cup U_2(t)$ be the corresponding event that round $t$ is unusual. If a round is \emph{not} unusual, then we say that it is \emph{regular}.
We also say that a given round is \emph{successful} if we move from phase $\ell$ to phase $\ell+1$. Let $S(t)$ be the event that round $t$ is successful. Similarly, round $t$ is called \emph{lucky} if we move from phase $\ell$ to phase $\ell+x$ for some $x \in \Nn \setminus \{1\}$, and $L(t)$ is the corresponding event.
(Let us stress the fact that in the definition of $S(t)$ and $L(t)$ we do \emph{not} condition that $U(t)$ does not hold. It will be important soon.)

Consider some round $t$ ($1 \le t \le N$) during some phase $\ell$ ($1 \le \ell \le k-1$). The probability that this round is unusual can be estimated as follows:
\begin{eqnarray*}
\Prob (U_1(t)) &\le& \binom{\ell}{2} \cdot \frac {1}{n-1} = O(k^2/n) \\
\Prob (U_2(t)) &\le& (k-\ell) \cdot \ell \cdot \frac {1}{n-1} = O(k^2/n),
\end{eqnarray*}
and so $\Prob(U(t)) = O(k^2/n)$. Hence, the expected number of unusual rounds is equal to
$$
\sum_{t=1}^N \Prob(U(t)) = O \left( N \ \frac {k^2}{n} \right) = O \left( n \ \frac {\omega \ln k}{k} \cdot \frac {k^2}{n} \right) = O (\omega k \ln k).
$$
It follows from Markov's inequality that a.a.s.\ the number of unusual rounds is at most $u := \omega^2 k \ln k$. We refer to this as property (P1).

There will be some unusual rounds but much fewer than the total number of rounds. Moreover, the possibility of an unusual round being successful or lucky is very low so that a.a.s.\ we will not see any unusual rounds being successful or lucky. We refer to this as property (P2). Indeed,
$$
\Prob( S(t) \cup L(t) ~|~ U(t) ) \le \ell \cdot (k-\ell) \cdot \frac {1}{n-1} = O(k^2/n).
$$
It follows that the expected number of unusual rounds that are successful or lucky is at most
\begin{eqnarray*}
\sum_{t=1}^N \Prob(U(t) \cap (S(t) \cup L(t)) ) &=& \sum_{t=1}^N \Prob(U(t)) \cdot \Prob(S(t) \cup L(t) ~|~ U(t)) \\
&=& N \cdot O(k^2/n) \cdot O(k^2/n) = O \left( \frac {\omega k^3 \ln k}{n} \right)  \\
&=& O \left( \frac {1}{\omega^2} \right) = o(1).
\end{eqnarray*}
We get that a.a.s.\ there is no such round by Markov's inequality. Similarly, a.a.s.\ there are no lucky rounds (regardless whether unusual or regular); We use property (P3) to refer to the property that
$$
\Prob (L(t)) \le \binom{k-\ell}{2} \cdot \left( \frac {\ell}{n-1} \right)^2 = O(k^4/n^2),
$$
and so the expected number of lucky rounds tends to zero as $n \to \infty$.

\medskip

Let us now condition on the events we know happen a.a.s., properties (P1)--(P3). We may couple the process with an auxiliary process where we simply ignore unusual rounds---these rounds are not successful nor lucky by property (P2). By property (P1), this auxiliary process has at most
$$
N - u = n \ \frac {\omega \ln k}{k} -\omega^2 k \ln k  \sim N
$$
rounds as $k = o( (n/\ln n)^{1/3} )$. If we reach that many rounds, we will stop the process prematurely. By property~(P3), no round is lucky and so we need to move from phase $1$ all the way to phase $k$ without skipping any phases. Since all rounds are regular and not lucky, at each round during phase $\ell$, we move to round $\ell+1$ with probability
\begin{eqnarray*}
p_{\ell} &:=& \Prob\left( \Bin \left(k-\ell, \frac {\ell}{n-1} \right) = 1 ~\Big|~ \Bin \left(k-\ell, \frac {\ell}{n-1} \right) \le 1 \right) \\
&:=& \frac {\Prob\left( \Bin \left(k-\ell, \frac {\ell}{n-1} \right) = 1 \right)}{\Prob\left( \Bin \left(k-\ell, \frac {\ell}{n-1} \right) = 1 \right)+\Prob\left( \Bin \left(k-\ell, \frac {\ell}{n-1} \right) = 0 \right)} \\
&=& \frac {(k-\ell) \cdot \frac {\ell}{n-1} \cdot \left( 1 - \frac {\ell}{n-1} \right)^{\ell-1}} {(k-\ell) \cdot \frac {\ell}{n-1} \cdot \left( 1 - \frac {\ell}{n-1} \right)^{\ell-1} + \left( 1 - \frac {\ell}{n-1} \right)^{\ell}}  \\
&=& \frac {\ell (k-\ell)}{n} \left( 1 + O(k^2/n) \right) \sim \frac {\ell (k-\ell)}{n};
\end{eqnarray*}
otherwise, we stay in phase $\ell$. It follows that the number of rounds it takes to move to the next phase is equal to $X_{\ell}$, the geometric random variable with expectation
$$
1/p_\ell = \frac {n}{(k-\ell)\ell} (1+O(k^2/n)) \sim \frac {n}{(k-\ell)\ell}.
$$
Hence, the length of the auxiliary process is $\max\{X, N-u\}$, where $X=\sum_{\ell=1}^{k-1} X_{\ell}$ is a sum of independent random variables $X_{\ell} = \Geom(p_\ell)$. It follows that
\begin{eqnarray}
\E [X] &=& \sum_{\ell=1}^{k-1} \frac {1}{p_\ell} = \sum_{\ell=1}^{k-1} \frac {n}{(k-\ell)\ell} (1+O(k^2/n)) \nonumber \\
&=& (1+O(k^2/n)) \ \frac {n}{k} \ \sum_{\ell=1}^{k-1} \left( \frac {1}{k-\ell} + \frac {1}{\ell} \right) \nonumber \\
&=& (1+O(k^2/n)) \  \frac {2n H_{k-1}}{k},\label{eq:expectation_X}
\end{eqnarray}
where $H_{k-1}=\sum_{\ell=1}^{k-1} 1/\ell$ is the Harmonic number. Since $H_{k-1} = \ln k + O(1)$, we get that
$$
\E [X] = (1+O(1/\ln k)) \frac {2n \ln k}{k}.
$$

\medskip

The rest of the proof is straightforward and follows from the concentration inequalities from~\cite{svante}. These bounds are obtained by the classical method of estimating the moment generating function (or probability generating function) and using the standard inequality (an instance of Markov's inequality).

It is proved in~\cite{svante} that for any $\lambda \ge 1$,
\begin{equation}
\Prob (X \ge \lambda \E[X]) \le \exp \Big( -p_* \cdot \E[X] \cdot (\lambda - 1 - \ln \lambda) \Big), \label{eq:Janson}
\end{equation}
where $p_* = \min_{\ell} p_{\ell} = (k-1)/n$. If $k=O(1)$, then one can take $\lambda = \sqrt{\omega}$ to get that
$$
\Prob (X \ge \lambda \E[X]) \le \exp \Big( - \Omega(\sqrt{\omega} \ln k) \Big) = o(1),
$$
and so a.a.s. the original process takes at most
$$
u + \min\{ \lambda \E[X], N-u \} \le \omega^2 k \ln k + (1+O(1/\ln k)) \frac {2 \sqrt{\omega} n \ln k}{k} < \frac {\omega n \ln k}{k} = N
$$
rounds. On the other hand, if $k \gg 1$, then one can take $\lambda = 1 + 1/(\ln k)^{1/3} \sim 1$ to get that
$$
\lambda - 1 - \ln \lambda \ge \left( 1 + \frac {1}{(\ln k)^{1/3}} \right) - 1 - \left( \frac {1}{(\ln k)^{1/3}} - \frac {1}{3(\ln k)^{2/3}} \right) = \frac {1}{3(\ln k)^{2/3}},
$$
since $\ln (1+x) = x - x^2/2 + O(x^3) \ge x - x^2/3$ for sufficiently small $x$. We get that
$$
\Prob (X \ge \lambda \E[X]) \le \exp \Big( - \Omega ( (\ln k) (\ln k)^{-2/3} ) \Big) = \Big( - \Omega ( (\ln k)^{1/3} ) \Big) = o(1),
$$
and so a.a.s.\ the original process takes at most
\begin{eqnarray*}
u + \max\{ X, N-u \} &\le& \omega^2 k \ln k + (1 + 1/(\ln k)^{1/3}) \ \E [X] \\
&=& (1 + O(1/(\ln k)^{1/3}))  \frac {2 n \ln k}{k}
\end{eqnarray*}
rounds. The upper bounds hold.

The lower bounds follow from upper bounds for the lower tail for a sum of geometric random variables. Indeed, it is proved in~\cite{svante} that for any $\lambda \le 1$,
$$
\Prob (X \le \lambda \E[X]) \le \exp \Big( -p_* \cdot \E[X] \cdot (\lambda - 1 - \ln \lambda) \Big).
$$
If $k=O(1)$, then one can take $\lambda = 1/\omega$ to get that
$$
\Prob (X \le \lambda \E[X]) \le \exp \Big( - \Omega( (\ln k) (-\ln (1/\omega)) ) \Big) = o(1).
$$
On the other hand, if $k \gg 1$, then one can take $\lambda = 1 - 1/(\ln k)^{1/3} \sim 1$ to get the same bound as before, namely,
$\lambda - 1 - \ln \lambda \ge \frac {1}{3(\ln k)^{2/3}}$, and so
$$
\Prob (X \le \lambda \E[X]) \le \exp \Big( - \Omega ( (\ln k) (\ln k)^{-2/3} ) \Big) = \Big( - \Omega ( (\ln k)^{1/3} ) \Big) = o(1).
$$
This finishes the proof.
\end{proof}

\subsection{Expectation for Case (a)}

Let us point out that we did not prove concentration for $\xi(K_n,k)$ in Case (a), that is when $k=O(1)$, but the bounds we did prove that hold a.a.s.\ are the best possible. Indeed, in order to illustrate this, we can consider the case $k=2$, though the same conclusion can be derived for any $k=O(1)$. As explained in the proof above, the length of the process can be modelled by the geometric random variable $X = \Geom(1/n)$ with $\E[X] = n$. Since for each $a \in (1,\infty)$ and $b \in (0,1)$ we have
\begin{eqnarray*}
\Prob ( X > a n) &=& \left( 1 - \frac {1}{n} \right)^{a n} \sim e^{-a}\\
\Prob ( X \le b n) &=& 1 - \Prob ( X > b n) = 1 - \left( 1 - \frac {1}{n} \right)^{b n} \sim 1 - e^{-b},
\end{eqnarray*}
no stronger bounds than the ones we proved in Theorem~\ref{thm:Kn:sparse} hold a.a.s.

\medskip

In order to finalize Case (a), let us compute the expected value of $\xi(K_n,k)$ for a constant $k$.
\begin{theorem}
If $k=O(1)$, then
$$
\E [\xi(K_n,k)] = \frac {2n H_{k-1}}{k} + O(\ln n ) \sim \frac {2n H_{k-1}}{k},
$$
where $H_{k-1}=\sum_{\ell=1}^{k-1} 1/\ell$ is the Harmonic number.
\end{theorem}

\begin{proof}
Suppose that at the beginning of some round $t$, there are $\ell$ green and $k-\ell$ white agents occupying $m \le k-\ell$ vertices; in particular, the process is at phase $\ell$. Let us first compute $a_{\ell} = a_{\ell}(t)$, the probability that the process stays at phase $\ell$, that is, during this round no white agent becomes green. We will use the notation and terminology introduced in the proof of Theorem~\ref{thm:Kn:sparse}. Recall that $U(t)^c$ is the event that this round is regular (not unusual), that is, during the initial sub-round no green agent moves to a vertex occupied by some white agent and no green agents meet. The white agents occupy $m$ vertices. By moving one green agent at a time, it is clear that
\begin{eqnarray*}
\Prob \left( U(t)^c \right) &=& \prod_{r = m}^{m+\ell-1} \left( 1 - \frac {r}{n-1} \right) \\
&=& 1 - \sum_{r = m}^{m+\ell-1} \frac {r}{n-1} + O \left( \frac {1}{(n-1)^2} \right) \\
&=& 1 - \frac {(2m+\ell-1)\ell}{2n} + O \left( \frac {1}{n^2} \right).
\end{eqnarray*}
Recall also that $S(t)$ and $L(t)$ are the events that this round is successful (that is, at the end of this round we move to phase $\ell+1$) and, respectively, lucky (that is, we move to phase $\ell+x$ for some $x \in \Nn \setminus\{1\}$). Let $b_{\ell} = b_{\ell}(t) := \Prob( S(t) )$ and $c_{\ell} = c_{\ell}(t) := \Prob( L(t) )$. Again, by moving one white agent at a time, we can easily compute the probability that we do not move to the next phase as follows:
\begin{eqnarray*}
\Prob \Big( (S(t) \cup L(t))^c ~|~ U(t)^c \Big) &=& \left( 1 - \frac {\ell}{n-1} \right)^j = 1 - \frac {\ell(k-\ell)}{n} + O \left( \frac {1}{n^2} \right) \\
\Prob \Big( (S(t) \cup L(t))^c ~|~ U(t) \Big) &=& 1 - O \left( \frac {1}{n} \right),
\end{eqnarray*}
and so
\begin{eqnarray*}
a_{\ell} &=& \Prob \Big( (S(t) \cup L(t))^c ~|~ U(t)^c \Big) \ \Prob \left( U(t)^c \right) + \Prob \Big( (S(t) \cup L(t))^c ~|~ U(t) \Big) \ \Prob \left( U(t) \right) \\
&=& \left( 1 - \frac {\ell(k-\ell)}{n} + O \left( \frac {1}{n^2} \right) \right) \left( 1 - \frac {(2m+\ell-1)\ell}{2n} + O \left( \frac {1}{n^2} \right) \right) \\
&& + \left( 1 - O \left( \frac {1}{n} \right) \right) \left( \frac {(2m+\ell-1)\ell}{2n} + O \left( \frac {1}{n^2} \right) \right) \\
&=& 1 - \frac {\ell(k-\ell)}{n} + O \left( \frac {1}{n^2} \right).
\end{eqnarray*}

Similarly, in order to compute $b_{\ell} $ and $c_{\ell} $, let us note that
\begin{eqnarray*}
\Prob \Big( S(t) ~|~ U(t)^c \Big) &=& (k-\ell) \cdot \frac {\ell}{n-1} \cdot \left( 1 - \frac {\ell}{n-1} \right)^{k-\ell-1} = \frac {\ell(k-\ell)}{n} + O \left( \frac {1}{n^2} \right) \\
\Prob \Big( S(t) ~|~ U(t) \Big) &=& O \left( \frac {1}{n} \right).
\end{eqnarray*}
It follows that
\begin{eqnarray*}
b_{\ell} &=& \Prob \Big( S(t) ~|~ U(t)^c \Big) \ \Prob \left( U(t)^c \right) + \Prob \Big( S(t) ~|~ U(t) \Big) \ \Prob \left( U(t) \right) \\
&=& \left( \frac {\ell(k-\ell)}{n} + O \left( \frac {1}{n^2} \right) \right) \left( 1 + O \left( \frac {1}{n} \right) \right) + O \left( \frac {1}{n} \right) \cdot O \left( \frac {1}{n} \right) \\
&=& \frac {\ell(k-\ell)}{n} + O \left( \frac {1}{n^2} \right),
\end{eqnarray*}
and so
$$
c_{\ell} =1-a_{\ell}-b_{\ell}=O \left( \frac {1}{n^2} \right).
$$

\medskip

To simplify the notation, let $X = \xi(K_n,k)$ be the number of rounds of the process. In order to get an upper bound for $\E[X]$, we couple the original process with an auxiliary process in which we move from phase $\ell$ to phase $\ell+1$ with probability
$$
\hat{b}_{\ell} = \ell (k-\ell) / n + O(1/n^2),
$$
and stay at phase $\ell$ otherwise. Let $Y$ be the number of rounds of this auxiliary process. The coupling ensures that $X \le Y$. More importantly, $Y=\sum_{\ell=1}^{k-1} Y_{\ell}$ is a sum of independent random variables $Y_{\ell} = \Geom(\hat{b}_\ell)$. Arguing as in the proof of Theorem~\ref{thm:Kn:sparse} (see~(\ref{eq:expectation_X})), we get that
$$
\E[X] \le \E [Y] = \frac {2nH_{k-1}}{k} + O(1),
$$
and the upper bound follows. Moreover, by applying~(\ref{eq:Janson}) we get that the contribution to $\E[Y]$ from rounds after, say, round $t_0 = t_0 (n) := 10 n \ln n$ is negligible, that is,
\begin{eqnarray}
\sum_{t \ge t_0} t \cdot \Prob(Y = t) &=& t_0 \cdot \Prob(Y \ge t_0) + \sum_{t \ge t_0+1} \Prob(Y \ge t) \nonumber \\
&\le& t_0 \cdot \exp \left( - \frac {(k-1) t_0}{n} (1+o(1)) \right) + \sum_{t \ge t_0} \exp \left( - \frac {(k-1) t}{n} (1+o(1)) \right) \nonumber \\
&\le& (10 n \ln n) \exp \left( - (10+o(1)) \ln n \right) + \sum_{t \ge t_0} \exp \left( - \frac {t}{n} (1+o(1)) \right) \nonumber \\
&\le& o(1) + \sum_{s \ge 10 \ln n} \sum_{t=ns}^{n(s+1)} \exp \Big( - s (1+o(1)) \Big) \nonumber \\
&=& o(1) + O(n) \cdot \sum_{s \ge 10 \ln n} \exp \Big( - s (1+o(1)) \Big) = o(1). \label{eq:tail}
\end{eqnarray}

In order to get a lower bound for $\E[X]$, we couple the original process with another auxiliary process. We stop this new auxiliary process prematurely at round $t$ if that round is lucky (that is, event $L(t)$ holds). Let $R(t)$ be the event that we stopped prematurely by round $t$; that is, $R(t) = \bigcup_{i=1}^t L(i)$. Let $Z$ be the number of rounds of this new auxiliary process. This time, the coupling ensures that $X \ge Z$. We get that
\begin{eqnarray*}
\E[Z] &=& \sum_{t \ge 1} t \cdot \Prob(Z = t) \\
&=& \sum_{t \ge 1} t \cdot \Big( \Prob(R(t)) + \Prob(R(t)^c) \cdot \Prob(Z = t ~|~ R(t)^c) \Big) \\
&\ge& \sum_{t = 1}^{10 n \ln n} t \cdot \Prob(R(t)^c) \cdot \Prob(Z = t ~|~ R(t)^c).
\end{eqnarray*}
Note that for any $t \le 10 n \ln n$,
$$
\Prob(R(t)^c) = \big( 1-O(1/n^2) \big)^t = 1 + O(t/n^2) = 1 + O(\ln n / n).
$$
Moreover, after conditioning on not finishing prematurely, we are back to the first auxiliary process, that is, $\Prob(Z = t ~|~ R(t)^c) = \Prob(Y = t)$. It follows that
\begin{eqnarray*}
\E[Z] &\ge& \Big( 1 + O(\ln n / n) \Big) \sum_{t = 1}^{10 n \ln n} t \cdot \Prob(Y = t) \\
&=& \Big( 1 + O(\ln n / n) \Big) \left( \sum_{t \ge 1} t \cdot \Prob(Y = t) + o (1) \right),
\end{eqnarray*}
by~(\ref{eq:tail}). We get that
$$
\E[Z] \ge \Big( 1 + O(\ln n / n) \Big) \left( \E[Y] + o (1) \right) = \frac {2n H_{k-1}}{k} + O(\ln n) \sim \frac {2n H_{k-1}}{k},
$$
and the proof is complete.
\end{proof}

\subsection{Concentration for the remaining subrange of Case (b)}

Let us now move to the case when $k$ is relatively large, that is, $k \ge n^{1/3} / \ln n$ but $k = o(n)$. (Note that Theorem~\ref{thm:Kn:sparse} requires that $k = o((n/\ln n)^{1/3})$ so the two theorems together cover the case $k=o(n)$, Case~(b).)

\begin{theorem}\label{thm:Kn:dense}
Let $\omega_0=\omega_0(n) \le \ln \ln n$ be any function tending to infinity as $n \to \infty$. Suppose that $k=k(n)$ is such that
$$
\frac {n^{1/3}}{\ln n} \le k \le \frac {n}{\omega_0}.
$$
Then the following property holds a.a.s.
$$
\xi(K_n,k) = \frac {2 n \ln k}{k} \left( 1 + O\left( \frac {1}{\omega_0} \right) \right) \sim \frac {2 n \ln k}{k}.
$$
\end{theorem}

\medskip

Before we prove this theorem, let us start with the following useful observation.

\begin{lemma}\label{lem:distr_of_agents}
Let $\omega_0=\omega_0(n)$ be any function tending to infinity as $n \to \infty$.
Let $s=s(n)$ be such that $\ln n \le s \le n / \omega_0$, and let $r=r(n) = 3s / \ln \omega_0 = o(s)$. Fix any $s$ agents (regardless whether they are green or white). The probability that they occupy at most $s-r$ vertices (in any given round in the future) is at most $1/n^2$.

In particular, it follows immediately by the union bound that the following property holds a.a.s.\ during the future $n$ rounds: the number of vertices occupied by the selected agents is more than
$$
s - O(r) = s (1+O(1/\ln \omega_0)) \sim s.
$$
(Note that, trivially, it is at most $s$. Moreover, different agents could be selected in each round and the number of them can vary, as long as they are selected before they actually make a move.)
\end{lemma}
\begin{proof}
The probability that the selected agents land on at most $s-r$ vertices is at most
\begin{eqnarray*}
\binom{n}{s-r} \left( \frac {s-r}{n-1} \right)^s &\le& \left( \frac {en}{s-r} \right)^{s-r} \left( \frac {s-r}{n} \right)^s (1+O(1/n))^s \\
&\le& e^{s-r+O(s/n)} \left( \frac {n}{s-r} \right)^{-r} \le e^{s} \left( \frac {n}{s} \right)^{-r} \\
&\le& \exp \left( s - r \ln \omega_0 \right) = \exp (-2s) \le n^{-2},
\end{eqnarray*}
as claimed.
\end{proof}

\medskip

Let us now come back to the main task, namely, bounding the length of the process. Due to the symmetry (that will be discussed in the proof below), we will concentrate on reaching phase $\ell=k/2$. Before we move to a formal argument, in order to build some intuition let us present a heuristic argument. Based on our experience so far, we expect that the number of rounds that are needed to move from phase $\ell=\ell_1$ to phase $\ell=\ell_2 \le k/2$, with $\ell_1 \ll \ell_2$, should be close to
\begin{eqnarray*}
\sum_{\ell=\ell_1}^{\ell_2-1} \frac {n}{(k-\ell)\ell} &=& \frac {n}{k} \ \sum_{\ell=\ell_1}^{\ell_2-1} \left( \frac {1}{k-\ell} + \frac {1}{\ell} \right) = \frac{n}{k} \left( O \left( \frac {\ell_2 - \ell_1}{k} \right) + H_{\ell_2-1} - H_{\ell_1-1} \right) \\
&=& \frac{n}{k} \Big( O(1) + \ln(\ell_2) - \ln(\ell_1) \Big) \sim \frac {n \ln(\ell_2/\ell_1)}{k}.
\end{eqnarray*}
On the other hand, the total number of rounds is asymptotic to $2 n \ln k / k = \Theta( n \ln n / k)$, since $k \ge n^{1/3}/\ln n$. Hence, if $\ell_2/\ell_1 = n^{o(1)}$, then the length of the part of the process between phase $\ell_1$ and $\ell_2$ is expected to be negligible. It should be stressed that this is not a formal argument, just a heuristic that suggests that such rounds are going to be negligible. Formal arguments will be provided below.

\medskip

For simplicity, we will distinguish a few stages of the process. Some of them (Stages 1, 2, and 4) will be difficult to control and so we will only manage to estimate the time they last up to a multiplicative constant. Fortunately, they will be negligible anyway. The other ones (Stage 3 and 5) are crucial and well-behaved.

Since we would like to provide one proof that covers the whole range of $k$, for some specific values of $k$ some stages actually do \emph{not} happen. It might be confusing at first so let us start with a brief discussion for each range of $k$. Stage~1 always happens and during this stage we reach $\ell = t_1 = \ln n$.
\begin{itemize}
\item $n^{1/3}/\ln n \le k \le \sqrt{2n/\ln n}$: We reach $\ell = t_2 = \ln^4 n$ at the end of Stage 2 and then finish with $\ell = t_3 = k/2$ at the end of Stage 3.
\item $\sqrt{2n/\ln n} < k \le \sqrt{2n \ln^2 n}$: We reach $\ell = t_2 = \ln^4 n$ at the end of Stage 2, $\ell = t_3 = n / (k \ln n)$ at the end of Stage 3, and then finish with $\ell = t_4 = k/2$ at the end of Stage 4.
\item $\sqrt{2n \ln^2 n} < k \le n / \ln^5 n$: We reach $\ell = t_2 = \ln^4 n$ at the end of Stage~2, $\ell = t_3 = n / (k \ln n)$ at the end of Stage~3, $\ell = t_4 = n \ln^2 n / k$ at the end of Stage~4, and then finish with $\ell = t_5 = k/2$ at the end of Stage~5.
\item $n / \ln^5 n < k \le n / \ln^2 n$: We reach $\ell = t_2 = n / (k \ln n)$ at the end of Stage~2, there is no Stage~3, we reach $\ell = t_4 = n \ln^2 n / k$ at the end of Stage~4, and then finish with $\ell = t_5 = k/2$ at the end of Stage~5.
\item $n / \ln^2 n < k \le n / \omega_0$: There is no Stage~2 nor Stage~3, we reach $\ell = t_4 = n \ln^2 n / k$ at the end of Stage~4, and then finish with $\ell = t_5 = k/2$ at the end of Stage~5.
\end{itemize}

\medskip

Finally, we are ready to move to the proof.

\begin{proof}[Proof of Theorem~\ref{thm:Kn:dense}]
Similarly to the proof of Theorem~\ref{thm:Kn:sparse}, since $k=o(n)$, we may assume that the process starts at phase 1.
As promised, we will distinguish a few stages of the process.

\medskip

\noindent \textbf{Stage 1}: This stage covers rounds until the number of green agents is at least $t_1 : = \ln n$, that is, when we reach phase $t_1$. Since $\ln t_1 = \ln \ln n = o(\ln n)$, the length of the process during this stage is expected to be negligible. As a result, in order to get a lower bound for $\xi(K_n,k)$ we simply ignore this stage, start the process with $t_1$ green agents, and couple such auxiliary process with the original one. If the auxiliary process is long so is the original one.

\medskip

In order to get an upper bound, we may terminate the process prematurely if it is not over by the end of round $N := 3 n \ln k / k$. As in the proof of Theorem~\ref{thm:Kn:sparse}, we split each round into two sub-rounds and let green agents move first. We observe that if the number of green agents stays below $t_1$, then green agents meet during the first $N$ rounds with probability at most
$$
N \cdot {t_1 \choose 2} \cdot \frac {1}{n-1} \le \frac {N t_1^2}{n} = \frac {3 t_1^2 \ln k}{k} \le \frac {3 \ln^{4} n}{n^{1/3}} = o(1).
$$
Since our process has to be finished after $N$ rounds (either naturally or prematurely), we may assume that no green agents meet during this stage of the process. On the other hand, since $k$ is large (recall that $k \ge n^{1/3}/\ln n$), when green agents move during the initial sub-round, they might move to a vertex occupied by some white agent (see event $U_2(t)$ defined in the proof of Theorem~\ref{thm:Kn:sparse}). If this happens, then it slightly slows the process down but we will show that it does so negligibly.

Consider any round during phase $\ell \le t_1$; there are $k-\ell \sim k$ white agents and $\ell \le t_1 = \ln n = o(k)$ green agents. By Lemma~\ref{lem:distr_of_agents}, we may assume that white agents always occupy $(k-\ell)(1+o(1)) \sim k$ vertices. Once they move, there are $(k-\ell)(1+o(1))$ white agents that do \emph{not} overlap with any green agent. If any of them moves to a vertex occupied by a green agent, this phase ends. It follows that the probability that the process stays at phase $\ell$ is at most
$$
q_\ell := \left( 1 - \frac {\ell}{n-1} \right)^{(k-\ell)(1+o(1))} \le \exp \left( - (1+o(1)) \frac {\ell(k-\ell)}{n} \right).
$$
Note that $e^{-x} \le 1 - 3x/5 < 1 - x/2$, provided that $x \in [0,1]$. Hence, if $\ell(k-\ell)\le n$, then
$$
q_\ell \le 1 - (1+o(1)) \frac {3\ell(k-\ell)}{5n} \le 1 - \frac {\ell(k-\ell)}{2n}.
$$

If $t_1(k-t_1) \le n$, then trivially $\ell (k-t_1) \le t_1(k-t_1) \le n$ and so the above bound for $q_\ell$ is always satisfied. In this case, we finish phase $\ell$ (and move to phase $\ell+x$ for some $x \in \Nn$) with probability at least $p_\ell \ge \ell(k-\ell)/(2n)$. Arguing as in the proof of Theorem~\ref{thm:Kn:sparse} we get that a.a.s.\ this stage takes at most
\begin{eqnarray*}
(1+o(1)) \sum_{\ell = 1}^{t_1-1} \frac {1}{p_\ell} &\le& (1+o(1)) \frac {2n}{k} \sum_{\ell = 1}^{t_1-1} \left( \frac {1}{k - \ell}  + \frac {1}{\ell} \right) \sim \frac {2n}{k} \left( \sum_{\ell = 1}^{t_1-1} \frac {1}{\ell}  + O(t_1/k) \right) \\
&\sim& \frac {2n}{k} \sum_{\ell = 1}^{t_1-1} \frac {1}{\ell} \sim \frac {2n \ln t_1}{k} = \frac {2n \ln \ln n}{k} \\
&=& O \left( N \ \frac {\ln \ln n}{\ln n} \right) = O(N/\omega_0) = o(N)
\end{eqnarray*}
rounds. Suppose then that $t_1(k-t_1) > n$, that is, $k$ is almost linear; in particular, $k \ge n / \ln n$. The argument above implies that a.a.s.\ we quickly reach phase $\ell_0$ for which $\ell_0(k-\ell_0) \ge n$. We will show now that we must reach the end of this phase in at most $4 \ln n$ additional rounds which is also negligible in comparison to $N$. Indeed, if $\ell(k-\ell) \ge n$, then
$$
q_\ell \le (1+o(1)) e^{-1} \le 1/2.
$$
It follows that during each round at least one white agent becomes green with probability at least $p_\ell \ge 1/2$. Hence, the expected number of white agents that turned green during $4 \ln n$ rounds can be stochastically bounded from below by the random variable $X \in \Bin(4 \ln n, 1/2)$. Since $\E[X] = 2 \ln n$, using Chernoff's bound~(\ref{eq:chern}) with $\eps=1/2$ we get that a.a.s.\ $X \ge t_1 = \ln n$. It follows that a.a.s.\ this phase will finish in at most $4 \ln n = O(N/\omega_0) = o(N)$ additional rounds.

\medskip

\noindent \textbf{Stages $\ge 2$}: For the remaining stages, we will continue using Lemma~\ref{lem:distr_of_agents}, which allows us to assume that during phase $\ell \ge t_1 = \ln n$ the number of vertices occupied by green agents is always more than $\ell (1-3/\ln \omega_0) \sim \ell$ (and, of course, at most $\ell$). It follows that any white agent becomes green with probability at least $\ell (1-3/\ln \omega_0)/(n-1) \sim \ell/n$ but at most $\ell/(n-1) \sim \ell/n$. It follows that the number of white agents that become green in one round is equal to
$$
Y_\ell \in \Bin \left( k-\ell, (1+o(1)) \ \frac {\ell}{n} \right)
$$
with $\E[Y_\ell] \sim (k-\ell)\ell/n$. Formally, in order to get an upper bound for $\xi(K_n,k)$, we need to couple the process using a sequence of random variables $\bar{Y}_\ell$ whereas for a lower bound we need to use $\hat{Y}_\ell$, where
$$
\bar{Y}_\ell \in \Bin \left( k-\ell, \frac {\ell}{n-1} \left( 1- \frac {3}{\ln \omega_0} \right) \right), \qquad \hat{Y}_\ell \in \Bin \left( k-\ell, \frac {\ell}{n-1} \right).
$$
However, in order to simplify the proof, we will use $Y_\ell$ instead of repeating the argument for both $\bar{Y}_\ell$ and $\hat{Y}_\ell$.

\medskip

The definition of random variables $Y_\ell$ does not change but it will still be convenient to distinguish some stages of the process depending on how large the expected value of $Y_\ell$ is. These stages will be treated differently. First, in order to make room for a technical argument, we need to reach $\ell=\ln^4 n$ in Stage~2. The length of this stage is not predictable but, since $\ln ((\ln^4 n) / (\ln n)) = 3 \ln \ln n = o(\ln n)$, we will show that it is negligible anyway. During Stage~3, the expected value of $Y_\ell$ is small, namely at most $1/\ln n$, so skipping phases is not common. The length of this stage can be well estimated. Once the expectation reaches $1/\ln n$ but is less than $\ln^2 n$, skipping phases may occur but the process is more challenging to analyze. Fortunately, since $\ln ((\ln^2 n) / (1/\ln n)) = 3 \ln \ln n = o(\ln n)$,  the length of this stage (Stage~4) will turn out to be negligible and so there is no need for a detailed analysis. Once we reach the expectation at least $\ln^2 n$, we reach Stage~5 when skipping phases becomes predictable and so the length of this stage is predictable too.

Finally, as we already mentioned, we stop the argument when more than $k/2$ agents become green. In the following argument, the only tool that we use is Chernoff's bound, which only depends on the expected value of the binomial random variable $Y_\ell$. Hence, due to the symmetry of the expected value of the binomial random variable $Y_\ell$, the second part of the process takes asymptotically the same amount of time. Indeed, if there are $\ell$ green agents, the number of white agents that become green is $Y_\ell \in \Bin(k-\ell, (1+o(1)) \ell/n)$ with $\E[Y_\ell] \sim (k-\ell)\ell/n$. On the other hand, if there are $\ell$ white agents, the number of white agents that become green is $Y_{k-\ell} \in \Bin(\ell, (1+o(1)) (k-\ell)/n)$ with $\E[Y_{k-\ell}] \sim (k-\ell)\ell/n$, as before.

\medskip

\noindent \textbf{Stage $2$}: This stage lasts until the number of green agents is at least
$$
t_2 := \min \left\{ \ln^4 n, \frac {n}{k \ln n} \right\}.
$$
As already mentioned, if $k > n / (\ln^2 n)$, then $t_2 \le \ln n = t_1$ and so it is possible that this stage actually does \emph{not} happen. Moreover, the length of this stage is negligible. We treat it independently since the number of green agents is still too small for the argument used in the next stage to be applied.

Note that $\ell(k-\ell) \le \ell k \le t_2 k \le n / \ln n \le n$. Arguing as in Stage~1, in each round we move to another phase with probability $p_\ell > \ell(k-\ell)/(2n)$ and so a.a.s.\ this stage finishes in at most $(2+o(1)) n \ln t_2 / k \le (8+o(1)) n \ln \ln n / k = O(N \ln \ln n / \ln n) = O(N/\omega_0) = o(N)$ rounds.

\medskip

\noindent \textbf{Stage $3$}: This stage lasts until the number of green agents is at least
$$
t_3 := \min \left\{ \frac {n}{k \ln n}, \frac {k}{2} \right\}.
$$
As already mentioned, if $k > n / \ln^5 n$, then $t_3 = t_2 = n / (k \ln n)$ and so it is possible that this stage does \emph{not} happen. If it does occur, then its length is asymptotically what we expect.

\medskip

Suppose that at some point of the process there are $\ell \ge \ln^4 n$ green agents. We will consider a chunk of
$$
r := \frac {n}{(k-\ell)\ell} \cdot \ln^3 n
$$
rounds but we stop the process prematurely if the number of green agents exceeds $\ell + 2 \ln^3 n$. Clearly, during this part of the process the number of green agents is equal to $\ell + O(\ln^3 n) \sim \ell$. As a result, since
$$
\ell (k-\ell) \le \ell k \le (t_3 + 2 \ln^3 n) k \le (1+o(1)) n / \ln n = o(n),
$$
the following properties hold during this part of the process:
\begin{eqnarray*}
\Prob (Y_\ell = 1) &\sim& (k-\ell) \cdot \frac {\ell}{n} \cdot \left( 1 - \frac {\ell}{n} \right)^{k-\ell-1} \sim \frac {(k-\ell)\ell}{n} =: w_1 \\
\Prob (Y_\ell \ge u) &\le& \binom{k-\ell}{u} \cdot \left( \frac {\ell}{n} \right)^{u} \le \left( \frac {e (k-\ell) \ell}{n u} \right)^{u} \le \left( \frac {2e k t_3}{n u} \right)^{u} \le \left( \frac {2e}{u \ln n} \right)^{u} =: w_u,
\end{eqnarray*}
for any $2 \le u \le \ln n$. In particular,
$$
\Prob (Y_\ell \ge \ln n) \le \left( \frac {2e}{\ln^2 n} \right)^{\ln n} \le (\ln n)^{-\ln n} = \exp( - (\ln \ln n) (\ln n) ) \le 1/n.
$$
Since we are only concerned with $N = O(n \ln k / k) = o(n)$ rounds, we may assume that $Y_\ell$ never exceeds $\ln n$.

The number of times we move from phase $\ell$ to $\ell+1$ can be modelled by random variable $X_1 \sim \Bin(r, w_1)$ with
$$
\E[X_1] = rw_1 \sim \ln^3 n.
$$
For a given $u \in \{2,3\}$, the number of times we move from phase $\ell$ to $\ell+i$, for some $i \ge u$, can be upper bounded by a random variable $X_{u} \sim \Bin(r, w_u)$ with
$$
\E[X_u] = rw_u = O(w_1^{u-1} \ln^3 n) = O(\ln^{4-u} n) = o(\ln^3 n).
$$
Hence, it follows from Chernoff's bounds that with probability at least $1-1/n$, $X_1 \sim \ln^3 n$, $X_2 = O(\ln^2 n)$, and $X_3 = O(\ln n)$.
It follows that with probability at least $1-1/n$,  the number of agents that become green during these $r$ rounds is asymptotic to
$$
X_1 + O(X_2) + O( (\ln n) \cdot X_3)  \sim \ln^3 n + O(\ln^2 n) + O( (\ln n) \cdot (\ln n)) \sim \ln^3 n.
$$
We conclude that we do not finish this chunk of rounds prematurely with probability at least $1-1/n$. Since there are $o(n)$ chunks of rounds (in fact, there are even only $o(n)$ rounds), a.a.s.\ we never finish prematurely. Moreover, note that it takes on average
$$
\frac {r}{(1+o(1)) \ln^3 n} \sim \frac {n}{(k-\ell)\ell} \ge \frac {n}{k t_3} \ge \ln n
$$
rounds to move from phase $\ell$ to $\ell+1$ so, indeed, the length of this stage is asymptotic to what one expects, namely, it is equal to $(1+o(1)) \sum_{\ell=t_2}^{t_3} n/((k-\ell)\ell)$.

\medskip

\noindent \textbf{Stage $4$}: This stage lasts until the number of green agents is at least
$$
t_4 := \min \left\{ \frac {n \ln^2 n}{k}, \frac {k}{2} \right\}.
$$
As already mentioned, if $k \le \sqrt{2 n / \ln n}$, then the process ends before we reach this stage. In any case, the length of this stage is negligible.

At the beginning of this stage, when $\ell(k-\ell) \le n$, we argue as in Stage~1 that the process moves to another phase with probability $p_\ell > \ell(k-\ell)/(2n)$. On the other hand, when $\ell(k-\ell) > n$, the expected number of agents that become green in one round is a binomial random variable $Y_\ell$ with $\E[Y_\ell] \sim \ell(k-\ell) / n > 1$. It follows from Chernoff's bounds (applied with $\eps = 1/3$) that $Y_\ell \ge (2/3) \E[Y_\ell] > \ell(k-\ell)/(2n)$ with probability at least $1-\exp(-1/27) > 1/30$. We get that the expected number of rounds in this Stage is at most
\begin{eqnarray*}
\sum_{\ell = t_3}^{t_4} \frac {30}{\ell(k-\ell)/(2n)} &=& 60 \sum_{\ell = t_3}^{t_4} \frac {n}{\ell(k-\ell)} \sim \frac {60 n \ln(t_4/t_3)}{k} = O \left( \frac {n \ln \ln n}{k} \right) \\
&=& O \left( N \ \frac {\ln \ln n}{\ln n} \right) = O(N/\omega_0) = o(N).
\end{eqnarray*}
It is straightforward to see that a.a.s.\ it is $o(N)$, as promised.

\medskip

\noindent \textbf{Stage $5$}: This stage lasts until the number of green agents reaches
$
t_5 := k/2.
$
As mentioned earlier, if $k \le \sqrt{2 n \ln^2 n}$, then the process ends before we reach this stage. On the other hand, if this stage occurs, then its length is predictable. Since
$$
\E[Y_\ell] \sim \frac {\ell (k-\ell)}{n} \ge \frac {\ell k}{2n} \ge \frac {t_4 k}{2n} = \frac {\ln^2 n}{2},
$$
it follows from Chernoff's bound (applied with, say, $\eps =\ln^{-1/3} n = o(1)$) that
$$
Y_\ell \sim \E[Y_\ell] \sim \frac {\ell(k-\ell)}{n}
$$
with probability at least $1-\exp(-\Theta(\ln^{4/3} n)) \ge 1-1/n$. Hence, a.a.s.\ $Y_\ell \sim \E[Y_\ell]$ during the whole stage. Suppose then that this is the case and it remains to compute the length of this stage. It is important to point out that $Y_\ell = \Theta (\ell k / n) = o(\ell)$ as then
$$
\sum_{i=\ell}^{\ell+Y_{\ell}-1} \frac {n}{i(k-i)} \sim \sum_{i=\ell}^{\ell+Y_{\ell}-1} \frac {n}{\ell(k-\ell)} = Y_{\ell}\ \frac {n}{\ell(k-\ell)} \sim 1.
$$
It follows that the length of this stage is asymptotic to what one expects, namely, it is equal to $(1+o(1)) \sum_{\ell=t_4}^{t_5} n/((k-\ell)\ell)$.

\medskip

Putting everything together we get that the total number of rounds until $k/2$ agents become green is a.a.s.\
\begin{align*}
o(N) + (1+o(1)) & \sum_{\ell=t_2}^{t_3} \frac {n}{(k-\ell)\ell} + (1+o(1)) \sum_{\ell=t_4}^{t_5} \frac {n}{(k-\ell)\ell} \\
& = o(N) + (1+o(1)) \sum_{\ell=1}^{k/2} \frac {n}{(k-\ell)\ell} \\
& = o(N) + (1+o(1)) \frac {n \ln (k/2)}{k} \sim \frac {n \ln k}{k}.
\end{align*}
By symmetry, as explained above, going from there to the end of the process, phase $k$, it takes asymptotically the same amount of time, thus concluding the proof.
\end{proof}

\subsection{Concentration for Case (c) and a subrange of Case (d)}
Let us now move to the situation where $k=k(n)$ is at least linear in $n$ but at most $n \ln^2 n$. As before, we will distinguish a few stages. Stage~2 and Stage~4 last for a non-negligible amount of time whereas Stage~1 and Stage~3 finish quickly and are negligible.

\begin{theorem}\label{thm:Kn:very_dense}
Let $\eps > 0$ be an arbitrarily small constant. Suppose that $k=cn$, where $c=c(n)$ is such that $\eps \le c \le \ln^2 n$.
Then the following property holds a.a.s.
$$
\xi(K_n,k) =  \left(1 + O\left( \frac {1}{\sqrt{\ln \ln \ln n}} \right) \right) \left(\frac {1}{\ln (1+c)} + \frac {1}{c} \right) \ln n \sim \left(\frac {1}{\ln (1+c)} + \frac {1}{c} \right) \ln n.
$$
\end{theorem}

Note that if $c = \Theta(1)$, then a.a.s.\ $\xi(K_n,k) = \Theta(\ln n)$ and both terms are of the same order. On the other hand, if $c = c(n) \to \infty$, then the second term is negligible compared to the first term and so a.a.s.\ $\xi(K_n,k) \sim \ln n / \ln c = o(\ln n)$.

\begin{proof}
Our goal is to show that a.a.s.\ it takes $\Theta(\ln n / \ln (1+c)) = \Omega(N)$ rounds to finish the process, where $N:=\ln n / \ln \ln n$. So stages that take $o(N)$ rounds to finish are negligible.

\medskip

\noindent \textbf{Stage 1}: This stage covers every round until the number of green agents is at least $t_1 : = \ln n$, that is, when we reach phase $t_1$. Arguing as in Theorem~\ref{thm:Kn:dense}, one can show that this stage a.a.s.\ takes $O(\ln \ln n) = O(N \cdot (\ln \ln n)^2 / \ln n) = o (N)$ rounds and so is negligible.
In fact, if (for example) $c \ge 1.1 \ln n$, then a.a.s.\ at least $\ln n$ agents start at the same vertex as the initial green agent and so this stage does not actually happen (that is, the required bound holds at round 0).
We omit the details.

\medskip

\noindent \textbf{Stage 2}: This stage covers every round before the number of green agents is at least $t_2 : = n / \ln \ln n$. Consider the beginning of some round at phase $\ell$, where $t_1 \le \ell < t_2$. As usual, we first move green agents and then white ones. Applying Lemma~\ref{lem:distr_of_agents} (with $\omega_0 = \omega_0(n) = \ln \ln n$) we may assume that when white agents make their move there are $\ell (1 + O(1/\ln \ln \ln n)) \sim \ell$ vertices occupied by green agents. As a result, the number of white agents that become green during this round can be modelled by the random variable $X \in \Bin (k-\ell, (1 + O(1/\ln \ln \ln n)) \ell / n)$ with the expectation equal to
$$
(k - \ell) \left( 1 + O \left( \frac{1}{\ln \ln \ln n} \right) \right) \frac {\ell}{n} = \frac {k \ell}{n} \left( 1 + O \left( \frac{1}{\ln \ln \ln n} \right) \right) = c \ell \left( 1 + O \left( \frac{1}{\ln \ln \ln n} \right) \right).
$$
It follows from Chernoff bound~(\ref{eq:chern}) (applied with $\eps=1/\ln \ln \ln n$) that
$$
X = c \ell (1 + O(1/\ln \ln \ln n))
$$
with probability at least
$$
1 - 2 \exp\left( \Omega \left( \eps^2 \E [X] \right) \right) = 1 - 2 \exp\left( \Omega \left( \frac {\ln n}{(\ln \ln \ln n)^2} \right) \right) \ge 1 - \frac {1}{\ln^2 n}.
$$
If this property holds, then we say that a given round is good. During each good round, the number of green agents increases from $\ell$ to
$$
\ell + c \ell (1 + O(1/\ln \ln \ln n)) = \ell (1+c) (1 + O(1/\ln \ln \ln n)).
$$

We will show that a.a.s.\ this stage takes $T$ rounds such that $T_- \le T \le T_+$, where
$$
T_{\pm} = \left( \log_{1+c} n \right) \left( 1 \pm \frac {1}{\sqrt{\ln \ln \ln n}} \right) = \frac {\ln n}{\ln (1+c)} \left( 1 \pm \frac {1}{\sqrt{\ln \ln \ln n}} \right) \sim \frac {\ln n}{\ln (1+c)}.
$$
Since the expected number of rounds that are not good is at most $T_+ / \ln^2 n = o(1)$, a.a.s.\ all rounds are good. But this implies that a.a.s.\ at the end of round $T_-$ the number of green agents is equal to
\begin{eqnarray*}
(\ln n) \left( (1+c) \left(1 + O\left( \frac {1}{\ln \ln \ln n} \right) \right) \right)^{T_-} &=& (1+c)^{T_-} \exp \left( \ln \ln n + O\left( \frac {T_-}{\ln \ln \ln n} \right) \right) \\
&=& n^{1-1/ \sqrt{\ln \ln \ln n}} \exp \left( O\left( \frac {\ln n}{\ln \ln \ln n} \right) \right) \\
&=& n \exp \left( - \frac {\ln n}{\sqrt{\ln \ln \ln n}}  + O\left( \frac {\ln n}{\ln \ln \ln n} \right) \right)  \\
&=& n \exp \left( - \frac {\ln n}{\sqrt{\ln \ln \ln n}} \ (1+o(1)) \right) \\
&<& \frac {n}{\ln \ln n} = t_2.
\end{eqnarray*}
So, indeed, this stage is not finished in less than $T_-$ rounds a.a.s. Similar calculations show that
$$
(\ln n) \left( (1+c) \left(1 + O\left( \frac {1}{\ln \ln \ln n} \right) \right) \right)^{T_+} = n \exp \left( \frac {\ln n}{\sqrt{\ln \ln \ln n}} \ (1+o(1)) \right) > n > t_2,
$$
and so a.a.s.\ this stage finishes in less than $T_+$ rounds.

\medskip

\noindent \textbf{Stage 3}: This stage covers every round before the number of green agents is at least $t_3 : = k - n / \ln n$. It will be easier to monitor the number of white agents. At the beginning of this stage, the number of white agents is at most $k - t_2 \le k \le n \ln^2 n$, and at the end of this stage it should be at most $n / \ln n$. Applying Lemma~\ref{lem:distr_of_agents} for agents that are green at the beginning of this stage (that is, with $s = s(n) = t_2 = n / \ln \ln n$ and $\omega_0 = \omega_0(n) = \ln \ln n$) we may assume that each time white agents move there are at least $s (1 + O(1/\ln \ln \ln n)) \sim s$ vertices occupied by green agents. This will be enough to show that the length of this stage is negligible. In other words, this part of the process is short even if the only way to become green is to meet an agent that is already green at the beginning of this stage.

\medskip

Suppose that at the beginning of some round, there are $w$ white agents, where $w > n / \ln n$. The number of white agents that become green at the end of this round can be stochastically lower bounded by random variable $X \in \Bin(w, (1+o(1)) / \ln \ln n)$ with
$$
\E[X] \sim \frac {w}{\ln \ln n} > \frac {n}{(\ln n) (\ln \ln n)}.
$$
It follows from Chernoff's bound~(\ref{eq:chern}) that $X > (2/3) \E[X] > w / (2 \ln \ln n)$ with probability at least $1 - 2\exp(-\Omega(\E[X])) \ge 1 - 1/n$. We will show that this phase will end in less than $T = 6 (\ln \ln n)^2 = O(N (\ln \ln n)^3 / \ln n) = o(N)$ rounds and so its length is negligible. Indeed, a.a.s.\ during these (at most) $T$ rounds, the number of white agents decreases each time by a multiplicative factor of at least $1-1/(2 \ln \ln n)$. It follows that the number of white agents after $t \le T$ rounds is at most
$$
(k-t_2) \left( 1 - \frac {1}{2 \ln \ln n} \right)^t \le (n \ln^2 n) \exp \left( - \frac {t}{2 \ln \ln n} \right).
$$
Since $(n \ln^2 n) \exp ( -T / (2 \ln \ln n) ) = (n \ln^2 n) \exp ( - 3 \ln \ln n ) = n / \ln n$, a.a.s.\ this phase has to finish in at most $T=o(N)$ rounds.

\medskip

\noindent \textbf{Stage 4}: We will continue the process till the very end. During this last stage, the number of green agents is equal to
$$
k - O(n/\ln n) = k(1-O(1/\ln n)) = cn (1-O(1/\ln n)) \sim cn.
$$
Our first task is to estimate the number of vertices occupied by them. It is a straightforward application of Chernoff's bound to get that a.a.s.\ no vertex will be occupied by, say, $O(\ln^2 n)$ agents during any round of that stage. Let us now concentrate on any vertex $v$. At the beginning of some round at phase $\ell$, there are $\ell_v = O(\ln^2 n)$ green agents occupying vertex $v$. Let $E_v$ be the event that no green agents moves to $v$. Clearly,
\begin{eqnarray*}
\Prob(E_v) &=& \left( 1 - \frac {1}{n-1} \right)^{\ell-\ell_v} = \left( 1 - \frac {1}{n-1} \right)^{cn (1+O(1/\ln n))-O(\ln^2 n)} \\
&=& \left( 1 - \frac {1}{n-1} \right)^{cn (1+O(1/\ln n))} = \exp \Big( - c (1+O(1/\ln n)) \Big).
\end{eqnarray*}

We will independently consider two cases. For small $c$ (Case~1), this stage has significant length and so we need to treat it carefully. If $c$ is large (Case~2), then its length is negligible and so some rough bound can be applied.

\medskip

\noindent \textbf{Case 1}: $c \le \sqrt{\ln n}$. Let $X = \sum_v I(E_v)$ be the number of vertices not occupied by any green agent ($I(E_v)$ is the indicator random variable for event $E_v$). It follows that
$$
\E[X] = n \exp( - c (1+O(1/\ln n))) \sim n e^{-c} = n^{1-o(1)}.
$$
Unfortunately, the events $E_v, E_w$ associated with vertices $v, w$ are not independent and so Chernoff's bound cannot be applied. However, they are almost independent and so it is straightforward to apply the second moment method to show the desired concentration. Indeed, for any pair of vertices $v,w$,
\begin{eqnarray*}
\Cov(I(E_v), I(E_w)) &=&  \Prob(E_v \cap E_w) - \Prob(E_v) \Prob(E_w) \\
&=& \left( 1 - \frac {2}{n-1} \right)^{\ell-\ell_v-\ell_w} \left( 1 - \frac {1}{n-1} \right)^{\ell_v} \left( 1 - \frac {1}{n-1} \right)^{\ell_w}\\
&& - \left( 1 - \frac {1}{n-1} \right)^{\ell-\ell_v} \left( 1 - \frac {1}{n-1} \right)^{\ell-\ell_w} \\
&=& \exp \left( - \frac {2\ell - \ell_v - \ell_w}{n-1} + O(\ell/n^2) \right) \\
&& - \exp \left( - \frac {2\ell - \ell_v - \ell_w}{n-1} + O(\ell/n^2) \right) \\
&=& \exp \left( - \frac {2\ell - \ell_v - \ell_w}{n-1} \right) \left( (1 + O(\ell/n^2))- (1 + O(\ell/n^2)) \right) \\
&=& O(\ell/n^2) \ e^{-2c}.
\end{eqnarray*}
We get that
\begin{eqnarray*}
\Var[X] &=& O(n^2) \cdot O(\ell/n^2) \ e^{-2c} = O(\ell/n^2) (\E[X])^2 \\
&=& O \left( \sqrt{\ln n}/n \right) (\E[X])^2 = o( (\E[X])^2 ).
\end{eqnarray*}
It follows from Chebyshev's inequality~(\ref{eq:cheb}) (applied with, say, $\eps=n^{-1/3}$) that with probability at least $1-1/\ln^2 n$, the number of vertices not occupied by any green agent is equal to $ n \exp( - c (1+O(1/\ln n)))$. We may then assume that this is the case during this stage of the process. It follows that each time a white agent moves, she stays white with probability $\exp( - c (1+O(1/\ln n)))$.

Let
$$
T_{\pm} = \frac {\ln n}{c} \left( 1 \pm \frac {2 \ln \ln n}{\ln n} \right) \sim \frac {\ln n}{c}.
$$
The probability that an agent that is white at the beginning of this stage stays white during $T_+$ rounds is equal to
$$
\exp \left( - c \left(1+O \left( \frac {1}{\ln n} \right) \right) \right)^{T_+} = \exp \left( - \ln n - 2 \ln \ln n + O(1) \right) = \Theta \left( \frac {1}{n \ln^2 n} \right).
$$
We get that the expected number of white agents at time $T_+$ is $o(1)$ and so a.a.s.\ we are done in at most $T_+$ rounds. On the other hand, the expected number of white agents at time $T_-$ is
$$
\frac {n} {\ln n} \exp \left( - c \left(1+O \left( \frac {1}{\ln n} \right) \right) \right)^{T_-} = \frac {n}{\ln n} \cdot \Theta \left( \frac {\ln^2 n}{n} \right) = \Theta( \ln n ) \to \infty.
$$
Chernoff's bound implies that a.a.s.\ this stage takes at least $T_-$ rounds, and the claimed bound holds.

\medskip

\noindent \textbf{Case 2}: $\sqrt{\ln n} := c_0 < c \le \ln^2 n$. Arguing as before, we may assume that each time a white agent moves, she stays white with probability at most $\exp( - c_0 (1+O(1/\ln n))) = \exp( -\sqrt{\ln n} + O(1/\sqrt{ \ln n}))$. The probability that there is at least one white agent left after $T = \sqrt{\ln n} = o(\ln n / \ln \ln n)$ rounds can be upper bounded as follows:
$$
\frac{n}{\ln n} \exp( -\sqrt{\ln n} + O(1/\sqrt{ \ln n}))^T = \frac {n}{\ln n} \cdot O \left( \frac {1}{n} \right) = o(1).
$$
We get that a.a.s.\ the length of this stage is negligible and the claimed bound holds.
\end{proof}

\subsection{Concentration for the remaining subrange of Case (d) and Cases (e)-(g)}

The situation when $k > n \ln^2 n$ is relatively easy to investigate. We will first deal with the case when $k = n^{o(1)}$. Since the proof is very similar (but much easier) to the one of Theorem~\ref{thm:Kn:very_dense} we provide only a sketch. After that it will be straightforward to finalize the remaining cases, Cases (e)-(g).

\begin{theorem}\label{thm:Kn:very_dense2}
Suppose that $k=cn$, where $c=c(n)$ is such that $\ln^2 n \le c = n^{o(1)}$.
Then the following property holds a.a.s.
$$
\xi(K_n,k) =  \left(1 + O\left( \frac {1}{\sqrt{\ln \ln \ln n}} \right) \right) \frac {\ln n}{\ln c} + O(1) \sim \frac {\ln n}{\ln c}.
$$
\end{theorem}
\begin{proof}[Sketch of the proof]
It follows immediately from Chernoff's bound that $(1+o(1)) c$ agents become green in round 0. Suppose that at the beginning of some round there are $\ell$ green agents, $(1+o(1)) \ln^2 n \le (1+o(1)) c \le \ell \le n / \ln \ln n$. By Lemma~\ref{lem:distr_of_agents}, we may assume that once they move, $\ell(1+O(1/\ln \ln \ln n))$ vertices are occupied by at least one green agent. Now, it is time for white agents to move. Arguing as in the proof of Theorem~\ref{thm:Kn:very_dense}, we may assume that at the end of this round, there are $\ell c (1+O(1/\ln \ln \ln n))$ green agents. After $(1+O(1/\sqrt{\ln \ln \ln n})) \ln n / \ln c$ rounds, the number of green agents is at least $n / \ln \ln n$.

The process will be over in at most two more rounds a.a.s. Indeed, by Lemma~\ref{lem:distr_of_agents}, we may assume that once green agents move there will be at least $(1+o(1)) n / \ln \ln n$ vertices occupied by at least one green agent. By Chernoff's bound, after white agents move there will be at least $(1+o(1)) c n / \ln \ln n$ green agents a.a.s. Moreover, after applying Chernoff's bound one more time, we get that a.a.s.\ no vertex is occupied by more than, say, $2c$ agents.

Let us consider any vertex $v$. The probability that no green agent arrives at this vertex is equal to
\begin{eqnarray*}
\left( 1 - \frac {1}{n-1} \right)^{ (1+o(1)) c n / \ln \ln n - O(c) } &\le& \exp \left( - (1+o(1)) \frac {c}{ \ln \ln n} \right) \\
&\le& \exp \left( - (1+o(1)) \frac {\ln^2 n}{ \ln \ln n} \right) = o \left( \frac {1}{n} \right).
\end{eqnarray*}
Hence, by the union bound, a.a.s.\ all vertices are occupied by at least one green agent and so the process is over once white agents move. The claimed bound holds and the proof is finished.
\end{proof}

\medskip

Let us point out that in the previous theorem, it is assumed that $c = n^{o(1)}$ and so $\ln c = o(\ln n)$. As a result, $\xi(K_n,k) \to \infty$ as $n \to \infty$.
If $k = n^{1+x+o(1)}$ for some $x > 0$ (Cases~(e)--(g)), then $\xi(K_n,k)$ does not tend to infinity anymore.

\medskip

Suppose first that $1/i < x < 1/(i-1)$ for some $i \in \Nn \setminus \{1,2\}$ (Case~(e)). The following properties hold a.a.s. The number of green agents at the end of round 0 is equal to $(1+o(1))c$, and then each round it keeps growing by a multiplicative factor of $(1+o(1))c$. It reaches
$$
(1+o(1)) c^{i} = n^{ix+o(1)} = n^{1 + i(x-1/i) +o(1)} \gg n \ln n
$$
at the end of round $i-1$. Arguing as before, at the beginning of round $i$ all vertices are occupied by at least one green and the process is over. It follows that a.a.s.\ $\xi(K_n,k) = i$.

\medskip

Suppose now that $x > 1/2$ (Case~(f)). Regardless of how large $x$ is, once $(1+o(1))c$ green agents move from vertex $v$ at the beginning of round 1, a.a.s.\ $(1+o(1))c$ white agents will move to $v$. Clearly, they will stay white at the end of round 1 since no green agent occupies $v$ at that point (deterministically). A.a.s.\ the process will end at round~2, and so $\xi(K_n,k) = 2$.

\medskip

Finally, suppose that $x = 1/i$ for some $i \in \Nn \setminus \{1\}$ (Case~(g)). A.a.s.\ at the end of round $i-1$, there are $(1+o(1)) c^i = (1+o(1)) (k/n)^i$ green agents. Our goal is to investigate random variable $X$, the number of vertices not occupied by any green agent at the beginning of round $i$. If $c^i = (k/n)^i < (1-\eps) n \ln n$, then
\begin{eqnarray*}
\E [X] &=& n \left( 1 - \frac {1}{n-1} \right)^{ (1+o(1)) c^i - O(c) } = n \exp \left( - \frac {(1+o(1)) c^i}{n} \right) \\
&\ge& n \exp \left( - (1+o(1)) (1-\eps) \ln n \right) = n^{\eps + o(1)} \to \infty,
\end{eqnarray*}
as $n \to \infty$. It is straightforward to see that a.a.s.\ $X>0$ and so a.a.s.\ the process needs one more round to finish. It follows that a.a.s.\ $\xi(K_n,k) = i+1$. On the other hand, if $c^i = (k/n)^i > (1+\eps) n \ln n$, then $\E[X] \le n^{-\eps+o(1)} \to 0$, as $n \to \infty$.  It follows that a.a.s.\ $X=0$ and so $\xi(K_n,k) = i$. One can obtain more precise results for the critical value when $c^i = (k/n)^i \sim n \ln n$ but we do not do so and only claim that a.a.s.\ $\xi(K_n,k) \in \{i, i+1\}$.

\subsection*{Acknowledgments}
This research was funded, in part, through a generous contribution from NXM Labs Inc.  NXM's autonomous security technology enables devices, including connected vehicles, to communicate securely with each other and their surroundings without human intervention while leveraging data at the edge to provide business intelligence and insights. NXM ensures data privacy and integrity by using a novel blockchain-based architecture which enables rapid and regulatory-compliant data monetization.

\newpage
\appendix

\section{Appendix: preliminary results of an empirical study}\label{sec:appendix}

\subsection{Model overview}

We are considering efficiency of message broadcasting in a dynamic transportation network.
We are representing the transportation network as an undirected planar graph $G=(V,E)$ having $n=|V|$ vertices and $m=|E|$ edges. At each point of time, there are $k \ge 2$ agents occupying the graph.

\begin{figure}
    \centering
    \includegraphics[width=0.8\textwidth]{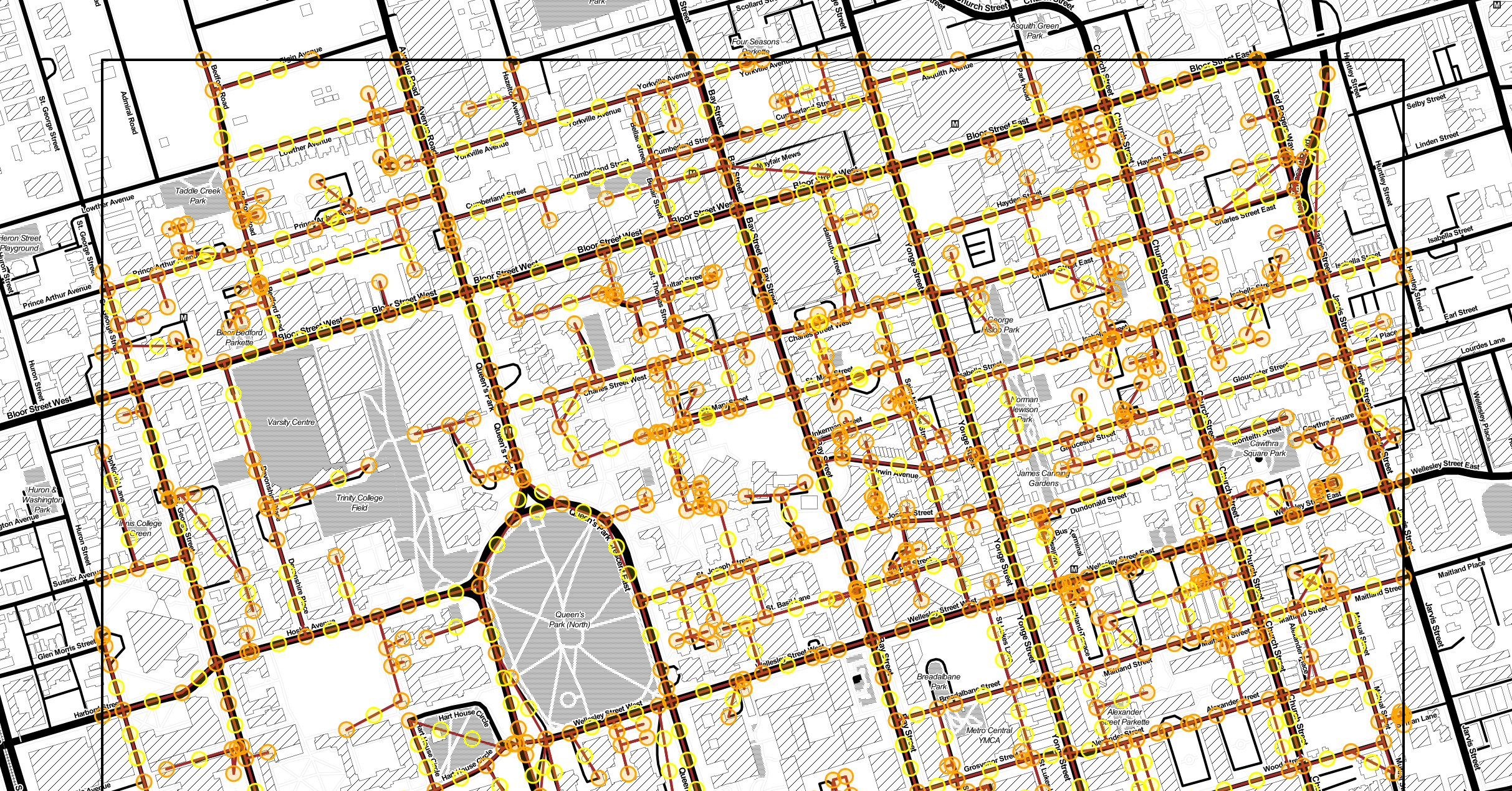}
    \caption{Graph discretization of central Toronto area with a maximum distance between graph vertices amounting to 50 meters. The vertices of the graph have two possible colours: orange and yellow. Orange vertices represent intersections from the Open Street Map data. Yellow vertices have been added during the discretization process to keep the maximum edge length below 50 meters. The edges in the graph have been presented with brown colour. The area of interest has been marked with black rectangle.  }
    \label{fig:torontonetwork}
\end{figure}

\begin{figure}
    \centering
    \includegraphics[width=0.8\textwidth]{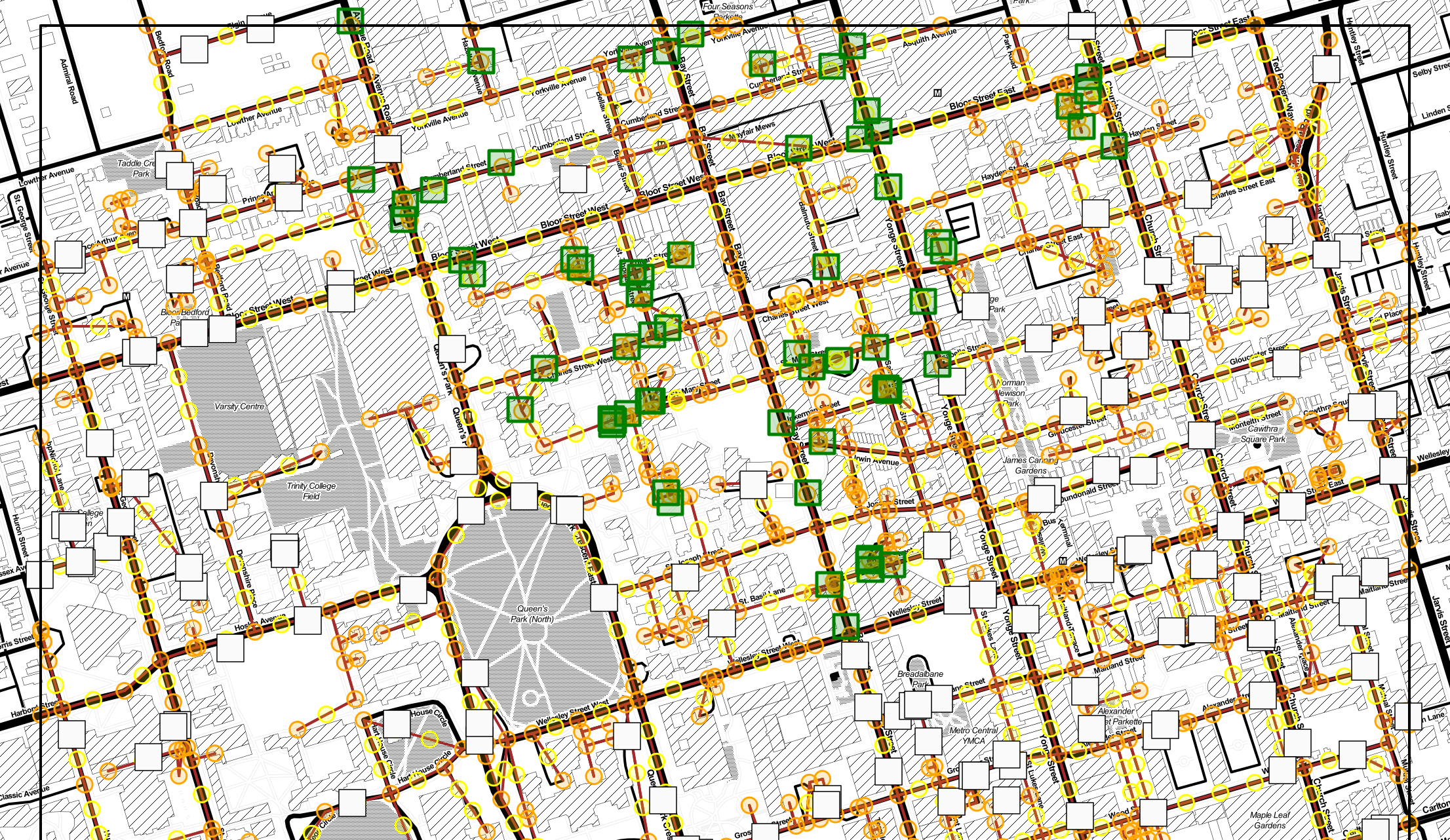}
    \caption{A sample simulation state after 100 steps. The agents represented by green rectangles have received the message while the white agents have not yet received it. Note that the picture presents only the top part of the simulated area. }
    \label{fig:torontoagents}
\end{figure}

We partition agents into two groups representing the following two possible states:
\begin{itemize}
    \item \emph{white agents}---have not received the message
    \item \emph{green agents}---have received the message.
\end{itemize}
At each step of the process, the agents change their location on the graph by moving to an adjacent vertex.
Independently of the colour of agents, the graph is also coloured and has two types of vertices:
\begin{itemize}
\item \emph{orange vertices}---where agents can randomly change their direction of travel
\item \emph{yellow vertices}---when going through such vertices, agents continue their travel in the same direction.
\end{itemize}
The degree of a yellow vertex is always $2$ while the degree of the orange vertex can be any natural number.

The model starts with all agents randomly located on vertices of the graph and one randomly selected agent broadcasting the message. If there are other agents occupying the same vertex the starting agent is located on, the message is passed to them.
At each step each agent moves to an adjacent vertex. More than one agent can occupy a vertex.
When agents move from an orange vertex they randomly select the next adjacent vertex with equal probability for each neighbour. On the other hand, when agents move from a yellow vertex they can only move to a vertex that they have not visited one step before (that is, they continue in the same direction).

There are two scenarios for passing messages in the model:
\begin{itemize}
    \item \emph{Jump-over = No}---agents pass the message only when they meet at the end of the step on the same vertex.
    \item \emph{Jump-over = Yes}---agents pass the message either when they meet at the end of the step on the same vertex or when during the step an agent having a message passes over an another agent.
\end{itemize}
Note that, if at the end of a step more than one agent occupies a vertex and one of agents has the message, then all agents that interacted will receive the message.

Our goal is to estimate how the number of agents, size, and the structure of the network determine the time required to broadcast the message to all agents presented in the model.

\subsection{Data for simulation: the central Toronto area}

We start by simulating the model of a real-world network, central Toronto. We start with a graph taken from the data from the OpenStreetMap project and represent each intersection in the road system as an orange vertex. Additionally, we convert the transportation network to an undirected graph. However, since in the presented model an agent moves to an adjacent vertex at each step and there is huge variance in intersection distances, we discretize the graph to make the model more realistic.
Namely, for some discretization parameter $d>0$, we divide each edge of length greater than $d$ into sub-edges. We achieve that by adding an appropriate, minimum number of yellow edges such that no edge in the resulting graph is longer than $d$.
Since yellow vertices have been added in the graph discretization process, each yellow vertex has exactly 2 neighbours. A sample graph after discretization has been presented in Figure~\ref{fig:torontonetwork}. The edges (marked with brown colour) are weighted; the weights represent distance between vertices in meters and are being used to control the parameter $d$.

Note that in Figure~\ref{fig:torontonetwork}, the \emph{orange vertices} come from the OpenStreetMap data and represent physical intersections in the city, while the \emph{yellow vertices} represent vertices that have been added in a map discretization process. In this figure $d=50$ (meters) value was used but other discretization parameters (25m and 75m) were also tested on the sane map (see Table~\ref{tab:discretization} for the list of considered discretizations and corresponding graph vertex and edge counts).  In any case, the yellow vertices were added in such a way that no edge in the graph was longer than the given parameter.

\begin{table}
\footnotesize
    \centering
    \begin{tabular}{lcc}
        Disretization level & vertices ($n$) & edges ($m$) \\
        \hline
        $25$ meters & 4251 & 4596\\
        $50$ meters & 2451 & 2796\\
        $75$ meters & 1881 & 2226\\

    \end{tabular}
    \caption{Discretization levels and the corresponding graph size.}
    \label{tab:discretization}
\end{table}

During the simulation process agents exchange information. See Figure~\ref{fig:torontoagents} for a sample simulation state after 100 steps. Some agents (marked with green colour) have already received the message, while others (marked with white colour) have not yet received it.

\begin{figure}
	\centering
	\includegraphics[width=0.5\textwidth]{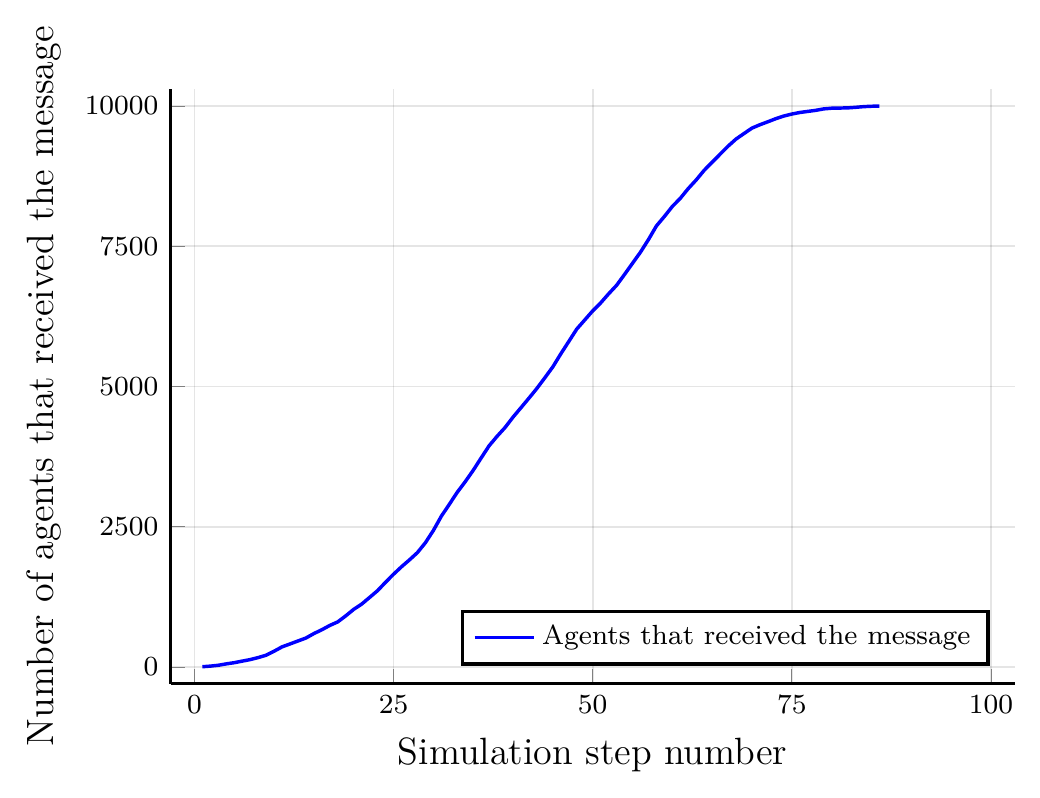}
	\caption{A sample simulation run for $k=10{,}000$ agents and a transportation network graph with a discretization level of $d=50$ meters.}
	\label{fig:samplerunsimple}
\end{figure}

We measured the number of agents who have received the message over time. For example, Figure~\ref{fig:samplerunsimple} presents a sample single simulation run for a population of $10{,}000$ agents. The simulation was started at the point of time $step=1$ where a single agent started to broadcast a message with other agents receiving it in subsequent steps. In the $Y$ axis we present the percentage of the population that has received the message at a given simulation step.

\begin{table}
\footnotesize
    \centering
    \begin{tabular}{l|r}
        Parameter & Considered values \\
        \hline
        discretization & [ $25$, $50$, $75$ ] \\
        jump-transmission & [Yes, No]   \\
        Number of agents $k$ & [ $10,40,70,100,150,200,$\\
        \quad &  $250,400,550,700,850,1000,$ \\
        \quad & $2000,3000,\ldots,10000,$\\
        \quad & $12000,14000,16000,18000,20000$ ]\\
        \hline
        Number of runs per each configuration & $3840$\\
        \hline
        \hline

    \end{tabular}
    \caption{Parameter sweep for the simulation experiments. A total of $ 3 \times 2 \times 26 = 156 $ parameter value sets have been considered. For each set of parameter values 3{,}840 simulations have been run. }
    \label{tab:parametersweep}
\end{table}

\subsection{Simulation results: message broadcasting dynamics }

In order to understand the determinants of message broadcasting dynamics, we have performed numerical simulations for various parameters of the model. Figures~\ref{fig:message_prop_med} and~\ref{fig:message_prop_high} show how the percentage of agents who received the message changes with the number of agents. Note that $90\%$ empirical confidence intervals are relatively wide. In practice this means that the observed variance of a single simulation run is quite large.

\begin{figure}
    \centering
    \includegraphics[width=0.65\textwidth]{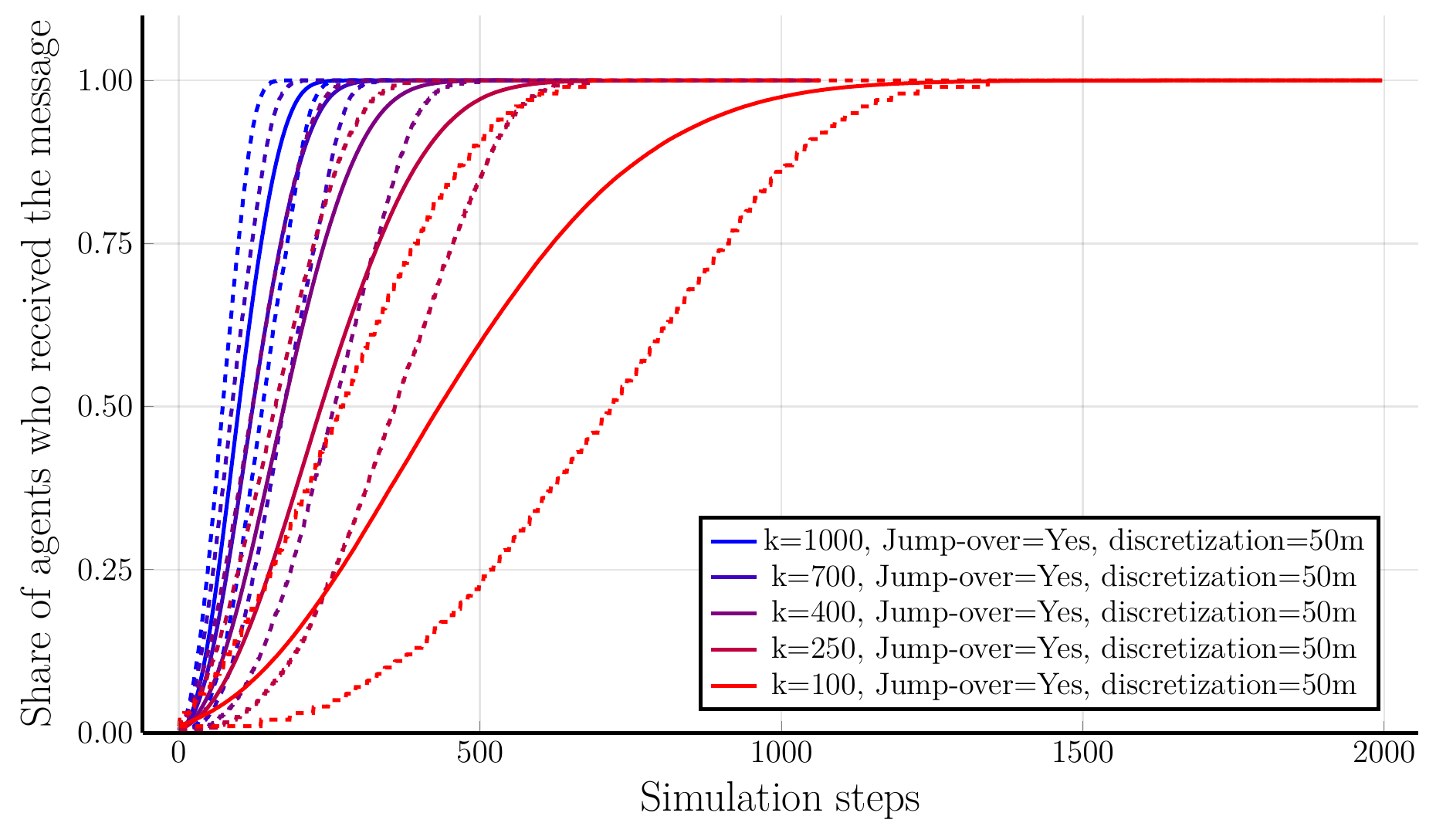}
    \caption{Message broadcasting dynamics for smaller agent populations. For each number of agents the results are averaged over 3,840 simulation runs. The dashed line represents the confidence interval for each observed values $90\%$ of simulation values for each configuration remain in that interval.}
    \label{fig:message_prop_med}
\end{figure}

\begin{figure}
    \centering
    \includegraphics[width=0.65\textwidth]{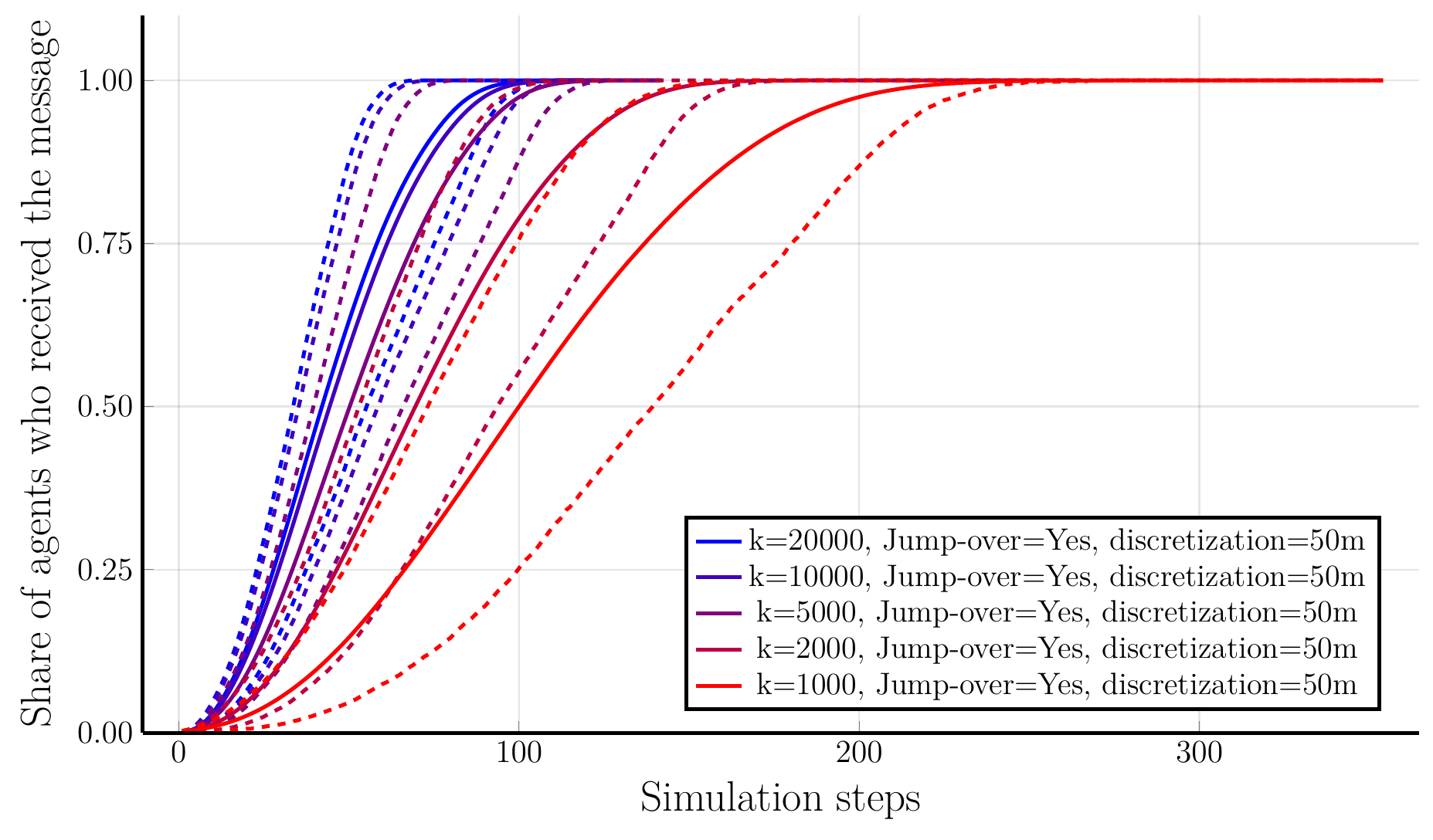}
    \caption{Message broadcasting dynamics for larger agent populations. For each number of agents the results are averaged over 3,840 simulation runs. The dashed line represents the confidence interval for each observed values, $90\%$ of simulation values for each configuration remain in that interval.}
    \label{fig:message_prop_high}
\end{figure}

\begin{figure}
    \centering
    \includegraphics[width=0.45\textwidth]{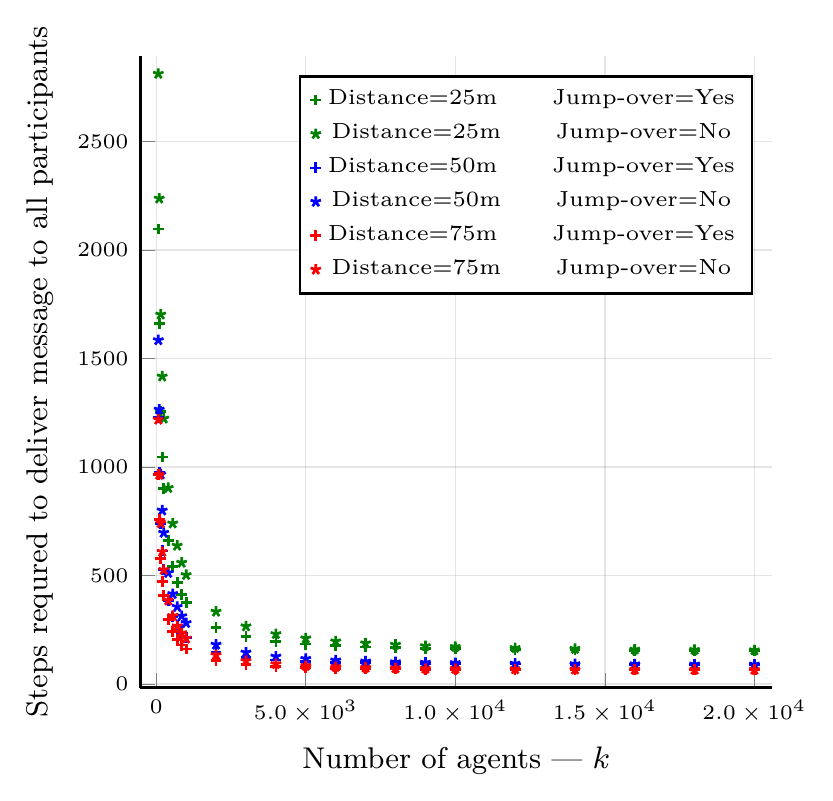}
    \caption{Number of steps required to broadcast a message to all system participant vs.\ the number of participants. For readability, only simulations for populations having at least $70$ agents are included.}
    \label{fig:agentsvssteps}
\end{figure}

Figure~\ref{fig:agentsvssteps} presents the number of steps required to broadcast a message to all system participant as a function of the number of participants. It can be seen that in larger agent populations and in more granular graphs the message is delivered much faster than for smaller population groups. Motivated by theoretical observations for complete graphs, we will scale the $X$ axis by $\ln k/k$ and observe the dependency.

\begin{figure}
    \begin{minipage}{0.70\linewidth}
        \centering
        \includegraphics[width=0.65\linewidth]{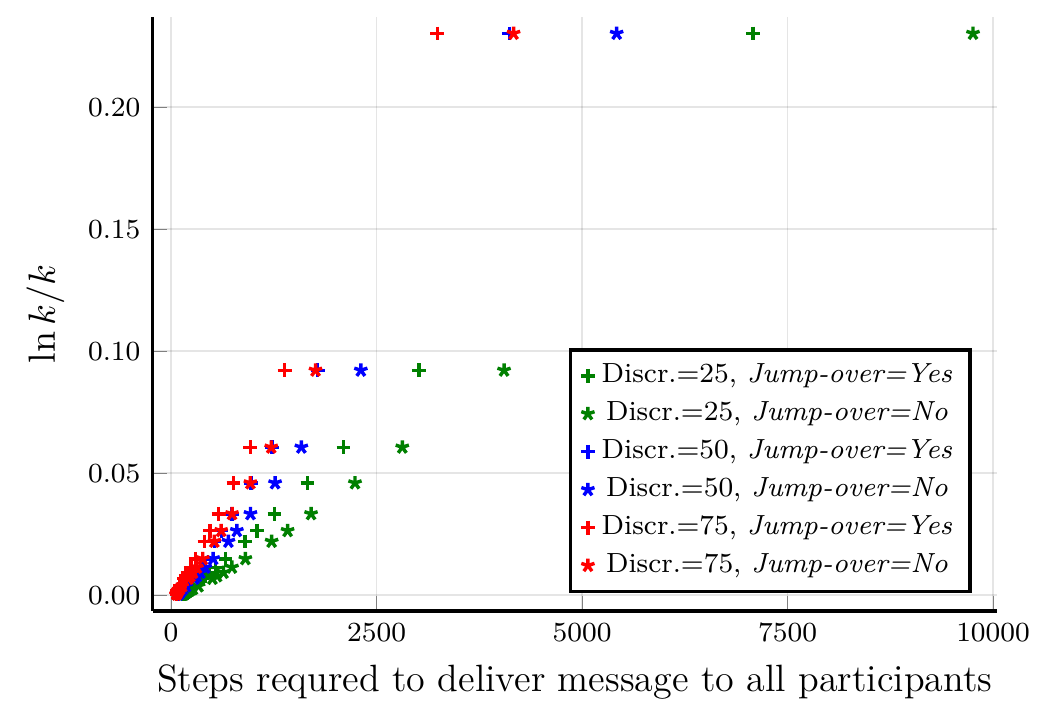}
    \end{minipage}%
    \begin{minipage}{0.2\linewidth}
        \centering
            \begin{tabular}{cc}
            Discret. & Corr.\\
            \hline
            \multicolumn{2}{c}{\textit{Jump-over=No}}\\
            $25$m & $0.99940$\\
            $50$m & $0.99917$\\
            $75$m & $0.99922$\\
            \hline
            \multicolumn{2}{c}{\textit{Jump-over=Yes}}\\
            $25$m & $0.99946$\\
            $50$m & $0.99928$\\
            $75$m & $0.99935$\\
            \hline
        \end{tabular}
    \end{minipage}
        \caption{Correlation of $\ln k/k$ to the actual number of steps required to broadcast the message for different graph discretization levels}%
    \label{fig:scatterplotlogk_over_k}
\end{figure}

\begin{figure}
    \begin{minipage}{0.70\linewidth}
        \centering
        \includegraphics[width=0.65\textwidth]{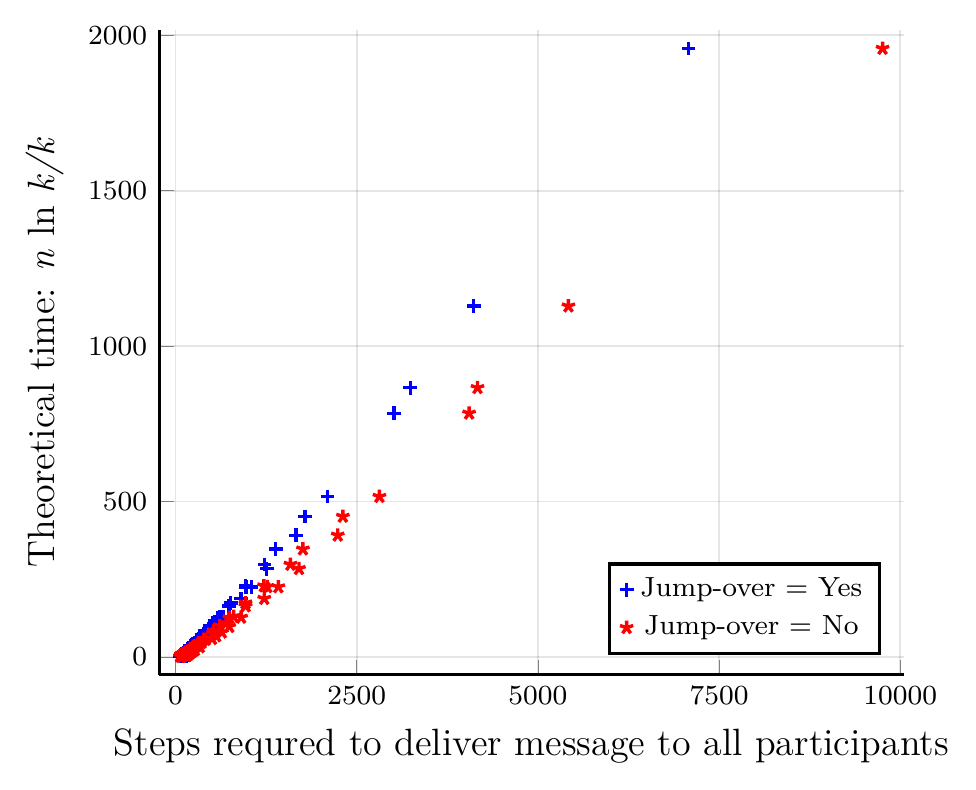}
    \end{minipage}%
    \begin{minipage}{0.3\linewidth}
        \centering
        \begin{tabular}{cc}
            \textit{Jump-over} & Corr.\\
            \hline
            Yes & $0.99879$\\
            No & $0.99879$\\
        \end{tabular}
    \end{minipage}
    \caption{There is an empirical linear dependency between $n \ln k/k$ and the number of steps required to complete the simulation. In the \emph{Jump-over = Yes} the observed time for fully broadcasting the message is shorter.}
    \label{fig:jump_yes_vs_no}
\end{figure}

It is clear from Figure~\ref{fig:scatterplotlogk_over_k} that there is a strong correlation between the expected number of steps required to deliver a message to all participants of the system (each point is an average for $3{,}840$ simulation runs) and the value of $\ln k/k$ where $k$ is the number of agents in the system.
We further aggregate the data across all discretization levels that correspond to different number of vertices in the graph (see Table~\ref{tab:discretization}). Since we expect the time required to travel to increase proportionally with the size of the graph, we multiply $\ln k/k$ by $n$ and hence measure the correlation of the empirical expected simulation times required to broadcast the message versus the theoretical values of $n\ln k/k$.
Again, the correlation values are around $0.999$ which suggest the existence of a strong dependency.

Finally, let us have a look at the aggregated data. We again consider two \emph{Jump-over} scenarios \emph{Yes} and \emph{No}---see Figure~\ref{fig:jump_yes_vs_no}. This time the simulation results for the entire parameter sweep (see Table~\ref{tab:parametersweep}) for all considered discretization levels (see Table~\ref{tab:discretization}) have been aggregated into one set. Each point in the figure is an average result of $3{,}840$ simulation runs.  Similarly to the previous results it is clear that including \emph{Jump-over} message passing reduces the required time by around $20$--$25$\%. The empirical results show that there is a strong correlation between $n \ln k/k$ and the number of steps required to fully broadcast a message. Note that in the presented scenarios the different values of $n$ taken into consideration depended on the discretization level while the structure of graph remained unchanged.


\begin{thebibliography}{99}

\bibitem{Wormald} H.\ Acan, A.\ Collevecchio, A.\ Mehrabian, and N.\ Wormald, On the push\&pull protocol for rumour spreading, SIAM J.\ Discrete Math.\ 31 (2017), 647--668.

\bibitem{voter2} D.\ Aldous, J.A.\ Fill, \emph{Reversible Markov Chains and Random Walks on Graphs}, 2002, unfinished monograph, available at \texttt{http://www.stat.berkeley.edu/$\sim$aldous/RWG/book.html}

\bibitem{Bai} F.\ Bai, D.\ D.\ Stancil, and H.\ Krishnan. Toward understanding characteristics of dedicated short range communications (DSRC) from a perspective of vehicular network engineers. In Proc. of the 16th annual international conference on Mobile computing and networking, pp. 329--340. 2010.

\bibitem{Pull4} S.\ Boyd, A.\ Ghosh, B.\ Prabhakar, and D.\ Shah. Randomized gossip algorithms. IEEE Transactions on Information Theory, 52(6):2508--2530, 2006.

\bibitem{Alan_d-reg} C.\ Cooper, A.\ Frieze, T.\ Radzik, Multiple random walks in random regular graphs, SIAM J.\ Discrete Math.\ \textbf{23(4)} (2009/10), 1738--1761.

\bibitem{Latin} R.\ Daknama, K.\ Panagiotou, S.\ Reisser, Asymptotics for Push on the Complete Graph, In 2020 Proceedings of the 14th Latin American Theoretical Informatics Symposium (LATIN 2020), accepted.

\bibitem{Pull5} A.\ Demers, D.\ Greene, C.\ Hauser, W.\ Irish, J.\ Larson, S.\ Shenker, H.\ Sturgis, D.\ Swinehart, and D.\ Terry. Epidemic algorithms for replicated database maintenance. In Proc.\ 6th Symp. Principles of Distributed Computing (PODC), pages 1--12, 1987.

\bibitem{Dey} K.C.\ Dey, A.\ Rayamajhi, M.\ Chowdhury, P. Bhavsar, and J. Martin. Vehicle-to-vehicle (V2V) and vehicle-to-infrastructure (V2I) communication in a heterogeneous wireless network–Performance evaluation. Transportation Research Part C: Emerging Technologies, 68, pp.168--184, 2016.

\bibitem{Push5} B.\ Doerr and M.\ K\"{u}nnemann. Tight analysis of randomized rumor spreading in complete graphs. In 2014 Proceedings of the Eleventh Workshop on Analytic Algorithmics and Combinatorics (ANALCO), pages 82--91. SIAM, 2014.

\bibitem{Push10} A.M.\ Frieze and G.R.\ Grimmett. The shortest-path problem for graphs with random arc-lengths. Discrete Applied Mathematics, 10(1):57--77, 1985.

\bibitem{svante} S.~Janson. Tail bounds for sums of geometric and exponential variables. Statistics Probability Letters 135 (2018), 1--6. 

\bibitem{JLR}  S.~Janson, T.~{\L}uczak, and A.~Ruci\'nski, {\em Random graphs}, Wiley, New York, 2000.

\bibitem{simulations} B.\ Kami\'nski, P.\ Pra\l{}at, and P.\ Szufel, On zombie infection in spatial city transportaion networks, working paper.

\bibitem{Lukasz} B.\ Kami\'nski, L.\ Krai\'nski, A.\ Mashatan, P.\ Pra\l{}at, and P.\ Szufel, Multi-agent routing simulation with partial smart vehicles penetration, Journal of Advanced Transportation, Volume 2020, Article ID 3152020 (2020), 11 pages.

\bibitem{Pull23} R.\ Karp, C.\ Schindelhauer, S.\ Shenker, and B.\ V\"{o}cking. Randomized Rumor Spreading. In Proc.\ 41st Symp.\ Foundations of Computer Science (FOCS), pages 565--574, 2000.

\bibitem{Kenney} J.B.\ Kenney. Dedicated short-range communications (DSRC) standards in the United States. Proceedings of the IEEE, 99(7), pages 1162--1182, 2011.


\bibitem{book_Yuval} D.A.\ Levin, Y.\ Peres, \emph{Markov Chains and Mixing Times} (Second Edition), AMS, 2017, 447 pp.

\bibitem{voter6} T.M.\ Liggett. Interacting Particle Systems, volume 276 of Grundlehren der Mathematischen Wissenschaften [Fundamental Principles of Mathematical Sciences]. Springer-Verlag, 1985.

\bibitem{Roberto12} R.\ Oliveira, On the coalescence time of reversible random walks, Transactions of the American Mathematical Society 364, no.\ 4 (2012): 2109--2128.

\bibitem{Roberto13} R.\ Oliveira, Mean field conditions for coalescing random walks, The Annals of Probability 41, no.\ 5 (2013): 3420--3461.

\bibitem{Push17} B.\ Pittel. On Spreading a Rumor. SIAM J.\ Appl.\ Math., 47(1):213--223, 1987.

\bibitem{Viriyasitavat} W.\ Viriyasitavat, O.K.\ Tonguz, and F.\ Bai. UV-CAST: an urban vehicular broadcast protocol. IEEE Communications Magazine 49, no. 11, pp 116--124, 2011.

\end{thebibliography}
\end{document}